\documentclass[12pt,draftcls,onecolumn]{IEEEtran}

\IEEEoverridecommandlockouts                              

\usepackage{empheq,amssymb}
\usepackage{braket}
\usepackage{dsfont,mathrsfs,bm}
\usepackage{enumerate}
\usepackage{float}
\usepackage{lipsum}
\usepackage{amsmath,amsthm}
\newtheorem{theorem}{Theorem}
\newtheorem{problem}{Problem}
\newtheorem{definition}{Definition}
\newtheorem{claim}{Claim}
\newtheorem{lemma}{Lemma}

\newtheorem{remark}{Remark}
\newtheorem{example}{Example}

\DeclareMathOperator*{\argmax}{arg\,max}
\let \VEC \mathbf

\long\def\comment#1{}

\let \lessthan \prec
\let \morethan \succ

\newcounter{l1}
\newcounter{l2}
\newcounter{l3}
\newcommand{\bdotlist}{\begin{list}{$\bullet$}{}}
\newcommand{\bboxlist}{\begin{list}{$\Box$}{}}
\newcommand{\bbboxlist}{\begin{list}{\raisebox{.005in}{{\tiny $\blacksquare$ \ \ }}}{}}
\newcommand{\bdashlist}{\begin{list}{$-$}{} }
\newcommand{\blist}{\begin{list}{}{} }
\newcommand{\barablist}{\begin{list}{\arabic{l1}}{\usecounter{l1}}}
\newcommand{\balphlist}{\begin{list}{(\alph{l2})}{\usecounter{l2}}}
\newcommand{\bAlphlist}{\begin{list}{\Alph{l2}.}{\usecounter{l2}}}
\newcommand{\bdiamlist}{\begin{list}{$\diamond$}{}}
\newcommand{\bromalist}{\begin{list}{(\roman{l3})}{\usecounter{l3}}}

\usepackage{tikz}
\usepackage{rotating}
\usepackage{pgfplots}
\pgfplotsset{compat=newest}
\pgfplotsset{plot coordinates/math parser=false}

\begin{document}

\title{
Duration-Differentiated Services in Electricity
}

\author{Ashutosh Nayyar,  Matias Negrete-Pincetic, Kameshwar Poolla and Pravin Varaiya
\thanks{The first author is with Ming Hsieh Department of Electrical Engineering at University of Southern California. The second, third and fourth authors are with the Department of Electrical Engineering and Computer Sciences, University of California, Berkeley. The second author is also with the Department of Electrical Engineering of Pontificia Universidad Catolica de Chile. This work was supported in part by EPRI and CERTS under sub-award 09-206; PSERC S-52; NSF under Grants 1135872, EECS-1129061, 
CPS-1239178, and CNS-1239274; the Republic of Singapore's National Research Foundation through a grant to the Berkeley Education Alliance for Research in Singapore for the SinBerBEST Program;
Robert Bosch LLC through its Bosch Energy Research Network funding program.
}}

\maketitle
\thispagestyle{empty}
\pagestyle{empty}

\begin{abstract}


The integration of renewable sources poses challenges at the operational and economic levels of the power grid. In terms of keeping the balance between supply and demand, the usual scheme of \textit{supply following load} may not be appropriate for large penetration levels of uncertain and intermittent renewable supply. In this paper, we focus on an alternative scheme in which the \textit{load follows the supply}, exploiting the flexibility associated with the demand side. We consider a model of flexible loads that are to be serviced by zero-marginal cost renewable power together with conventional generation if necessary.  Each load demands 1 kW for a specified number of time slots within an operational period. The flexibility of a load resides in the fact that the service may be delivered over any slots within the operational period.  Loads therefore require flexible energy services that are differentiated by the demanded duration. We focus on two problems associated with durations-differentiated loads. The first problem deals with the operational decisions that a supplier has to make to serve a given set of duration differentiated loads. The second problem focuses on a market implementation for duration differentiated services. We give necessary and sufficient conditions under which the available power can service the loads, and we describe  an algorithm that constructs an appropriate allocation. In the event the available supply is inadequate, we  characterize the minimum amount of power that must be purchased to service  the loads.
Next we consider a forward market where consumers can purchase duration differentiated energy services.  We first characterize social welfare maximizing allocations in this forward market and then show the existence of  an efficient competitive equilibrium. We also investigate the competitive equilibrium  in a sequence of  real-time spot markets with flexible loads. We show by an example that the sequence of real-time markets may not be efficient.

\end{abstract}

\section{Introduction} \label{sec-introduction}
The worldwide interest in renewable energy such as wind and solar is driven by pressing environmental problems, energy supply security and nuclear power safety concerns. The energy production from these renewable sources is {\em variable}: uncontrollable, intermittent, uncertain. 
Variability is a challenge to deep renewable integration.

A central problem is that of economically balancing demand and supply of electricity in the 
presence of large amounts of variable generation. The current {\em supply side} approach 
is to absorb the  variability in operating reserves.
Here, renewables are treated as negative demand, so the variability appears as uncertainty in net load which is compensated by scheduling  fast-acting reserve generation.  This strategy of {\em tailoring supply to meet demand} works
at today’s modest penetration levels. But it will not scale. Recent studies in California, e.g., \cite{CAISO},  project that the load-following capacity requirements will need to increase from 2.3 GW to 4.4 GW. These large increases in reserves will significantly raise electricity cost, and diminish the net carbon benefit from renewables as argued in several papers in the literature (\cite{kirschen2010,negwankowshamey12}).

There is an emerging consensus that demand side resources must play a key role in supplying zero-emissions regulation services that are necessary for deep renewable integration  (e.g., \cite{cal09K, galus2010, papaoren2010, matdyscal12K, anand2012}). These include thermostatically controlled loads (TCLs), electric vehicles (EVs), and smart appliances.  Some of these loads are deferrable: they can be shifted over time.  For example, charging of electrical vehicles (EVs) may be postponed to some degree.    Other loads such as HVAC units can be modulated within limits.  
The core idea of {\em demand side approaches} to renewable integration is to exploit load flexibility to 
track variability in supply, i.e., to tailor demand to match supply.  For this a cluster manager or aggregator 
offers a control and business interface between the loads and the system operator (SO). 

The demand side approach has led to two streams of work: (a) indirect load control (ILC) where flexible loads respond, in real-time, to price proxy signals, and (b) direct load control (DLC) where flexible loads cede physical control of devices to operators who determine appropriate actions. The advantage of DLC is that with greater control authority the cluster manager can more reliably control the aggregate load.  
However, DLC requires a more extensive control and communication infrastructure and the manager must provide economic incentives to recruit a sufficient consumer base. The advantage of ILC is that the consumer retains authority over her electricity consumption.  


Both ILC and DLC require appropriate economic incentives for the consumers.  In ILC, the real-time price signals provides the required incentives. However, the quantification of those prices, the feasibility of consumer response and the impact on the system and market operations in terms of price volatility and instabilities is a matter of concern, as recent literature suggest (\cite{RoozbehaniCDC2010,LBNL-RTPReport,wannegkowshameysha11b}).  

DLC also requires the creation of economic signals, but unlike real-time pricing schemes DLC can  use forward markets. For DLC to be effective, it is necessary to offer consumers who present greater demand flexibility a larger discount. The discounted pricing can be arranged through flexibility-differentiated electricity markets. Here, electricity is regarded as a set of differentiated services as opposed to a homogeneous 
commodity.  Consumers can purchase an appropriate bundle of services that best meets their electricity needs.
From the producer's perspective, providing differentiated services may better accommodate supply variability. 
This paper is concerned with electric power services differentiated by the duration $h$ for which power is supplied. 
We explore balancing supply and demand for such services through forward markets.
There is a growing body of work \cite{tanvar93, pravin2011,negmey12, bitar2013} on differentiated electricity services.

\subsection{Prior Work} \label{sec-prior}

\subsection*{Supply side approaches} 

Here, variable supply from renewable sources is regarded as negative demand.
The objective is to arrange for reserve power generation to compensate for fluctuations in net demand.  
The problem is formulated from the viewpoint of the system operator (SO)
who must purchase reserve generation capacity and energy to meet the random demand while 
minimizing the risk of mismatch and the cost of reserves. 
Reserve generation can be purchased in forward markets with different time horizons (day-ahead, hour-ahead, 5 minutes-ahead).
With shorter time horizons, the uncertainty in the forecast net demand is reduced but the cost of reserves increases.  
The SO's  optimal decision can be formulated as a stochastic control problem known as risk-limiting dispatch presented in \cite{pravin2011}.  
When SO's decisions include unit commitment  and transmission constraints, the problem 
is a mixed-integer nonlinear stochastic programming problem that is computationally challenging as discussed in \cite{galiana}. A number of papers address the
computational aspects of stochastic unit commitment (e.g., \cite{carpentier1996, nowak2000, paporeone2011}). Alternatively, \cite{berlitsun2013} present a robust optimization formulation of the unit commitment problem. If unit commitment and transmission constraints are omitted, the resulting stochastic dispatch problem has an analytical solution as shown in \cite{varaiya_rld, PMAPS2012}. 

\subsection*{Demand side approaches}

Current research in {\em direct load control}  focuses on developing and analyzing algorithms for coordinating resources (e.g., \cite{ChenLLF2011, LeeTPS2008, HsuDLC1991, MetsEVCharging2010}). For example, \cite{GanPESGM2012} develops a distributed scheduling protocol for electric vehicle charging; \cite{papaoren2010} uses approximate dynamic programming to couple wind generation with deferrable loads; and \cite{MaIFAC2011, GHGIREP2010, galus2010} suggest the use of receding horizon control approaches for resource scheduling.

Recent studies in {\em indirect load control} have developed real-time pricing algorithms \cite{IlicTPSNov2011} and
quantified operational benefits \cite{LijesenRTPElasticity}.  There has also been research focused on economic efficiency in \cite{Borenstein2005,Spees2007}, feedback stability of price signals in \cite{RoozbehaniCDC2010}, volatility of real-time markets in \cite{wannegkowshameysha11b} as well as the practical issues associated with implementing ILC programs presented in \cite{LBNL-RTPReport}. 

An early exposition of differentiated energy services is offered in \cite{OrenSmith}.
There are other approaches to such services that naturally serve to integrate variable generation sources. 
Reliability differentiated energy services where consumers accept contracts for $p$ MW of power with probability $\rho$ are developed in \cite{tanvar93}. More recently, the works of \cite{bitarlow2012, bitar2013} consider deadline differentiated contracts where consumers receive price discounts for offering larger windows for the delivery of $E$ MWh of energy. 

\subsection{Main Contributions}

In this paper, we consider a class of flexible loads that require a fixed power level for a specified duration within an operational period. The loads are differentiated based on the duration of service they require, we refer to them as duration-differentiated (DD) loads. We consider a stylized version of DD loads. The service interval is divided into $T$ slots, indexed $t = 1, \cdots, T$. A flexible load demands 1 kW of power for a duration of $h$ slots.  While this abstraction does not account for many important practical constraints, it serves to formulate and study the central mathematical problems in scheduling/control and markets for DD loads. The flexibility of a load resides in the fact that {\em any} $h$ of the available $T$ slots will satisfy the load. Examples of DD loads include electric vehicles that allow flexible charging over an 8 hour service interval, aluminum smelters that might operate for $h$ hours out of $24$, appliances such as washing machines that require a fixed power for any $1$ hour out of the next $8$. 



Our objective is to study the allocation of available supply to the various loads in a market context. 
We assume that if the supplier contracts {\em ex ante} to deliver power for $h$ slots to a particular flexible load, it is obligated to do so. The supplier selects $h$ of the the available $T$ slots to supply power to the load. This scheme requires certain technology infrastructure (communication, power electronics), a treatment of which is outside the scope of this paper. The load is not informed much in advance which $h$ slots it will receive power. Thus the load must assume the burden of planning its consumption without knowing exactly when power will be available. The available power is drawn from zero-marginal  cost renewable sources as well as electricity purchases made by the supplier in the day-ahead. Because of the variability in renewable sources, the supplier may be compelled to use supplemental generation such as on-site gas turbine or buying from a real-time market to meet its obligations.


Our first set of results are contained in Section \ref{sec:discrete}. Here, we study the decision problems faced by a supplier who has to serve a fixed set of DD loads.  The basic question we address is the following: given a forecast model of its renewable generation, what day-ahead and real time power purchases should a supplier make to ensure that all loads are served at least possible cost? To solve this problem, 
we first give a necessary and sufficient condition under which the available power $\mathbf{p}= (p_1, p_2, \ldots,p_T)$ is adequate to meet the loads. We describe a Least Laxity First (LLF) algorithm that constructs an appropriate allocation to serve the loads.   
In the event the available supply profile $\mathbf{p}$ does not meet the adequacy condition, we characterize the minimum cost power purchase decisions under (a) oracle information, and (b) run-time information about the supply. We use this solution to construct optimal day-ahead and real-time decisions for the supplier in Section~\ref{sec:supplier_problem}. 

Our second set of results may be found in Section \ref{sec:market}. Here, we consider a stylized {\em forward market for duration-differentiated services}, which are bundles of  $h$ 1-kW slots sold at prices $\pi(h)$, $h=1,2, \cdots$.  
Consumer $i$ select the service $h$ that maximizes her net utility $U(h) - \pi(h)$,
and the supplier bundles its supply  (both its renewable generation and any forward purchases made from the grid) into 
$n_h$ units of service for $h$ slots, so as to maximize its revenue.  We show that there is a competitive market equilibrium which maximizes social welfare.  The competitive duration-differentiated market equilibrium is then compared with a sequential real-time market, 
in which the price of power $\zeta (t)$ is the market clearing price for slot $t$.  
The comparison reveals that the real-time markets may not be efficient.
All proofs are collated in the Appendix. Concluding remarks and future research avenues are discussed in Section \ref{sec-conclusions}.


\subsection*{Notation}

Bold letters denote vectors. We reserve subscripts $t$ to index time and $i$ to index loads.
For a vector $\VEC a =(a_1,a_2,\ldots,a_T)$, $\VEC a^{\downarrow}$ denotes the non-increasing 
rearrangement of $\VEC a$, so, $a^{\downarrow}_t \geq a^{\downarrow}_{t+1}$ for $t=1,2,\ldots,T-1$.
For an assertion $A$, $\mathds{1}_{A}$ denotes $1$ if $A$ is true and $0$ if $A$ is false.



%

\section{Problem Statement}\label{sec:PF}
We formalize the two problems investigated in this work. The first problem deals with the operational decisions that a supplier has to make to serve a given set of duration differentiated loads. The second problem focuses on a market implementation for duration differentiated services.


\subsection{Serving a collection of duration differentiated loads}
We consider a discrete-time setting where time is segmented into $T$ slots, indexed by $t$.   The power available in slot $t$ is assumed to take non-negative integer values. 
There are  $N$ flexible loads, indexed by $i$. Load $i$ requires $1$ kW for {\em any}  $h_i$  of $T$ time slots.  The vector $\VEC h = (h_1,h_2,\ldots,h_N)$ is called the \emph{demand profile}.  For a demand profile $\VEC h$, we define an associated \emph{demand-duration vector} $\VEC d = (d_1,d_2,\ldots,d_T)$ where $d_t := \sum_i\mathds{1}_{\{t \leq h_i\}}$. Note that the number of consumers that need service for $t$ slots is  $d_t - d_{t+1}$ (where $d_{T+1} := 0$). Thus, there is a bijection between the demand profile and the demand-duration vector as $\VEC h$ specifies $\VEC d$ uniquely and {\em vice versa}.  

We consider the supplier's problem of serving a given collection of these  loads. The supplier owns renewable energy resources that can provide free but uncertain power. In addition, the supplier can purchase power in the day-ahead and real-time electricity markets. The renewable generation over the delivery period is denoted by the  non-negative integer-valued random vector  $\VEC R = (R_1,R_2,\ldots,R_T)$; the realizations of this vector are denoted by $\VEC r$. In the day-ahead market, the supplier has a probability mass function on this random vector given by $f(\VEC r)$. If the supplier purchases a power profile $\VEC y = (y_1,y_2,\ldots,y_T)$ in the day-ahead market, its cost is given by $c^{da}(\sum_{t=1}^T  y_t)$, where $c^{da}$ is day-ahead market price. In real-time, the random vector of renewable generation takes a realization $\VEC r$, so that the total realized supply profile is $\VEC p = \VEC r+ \VEC y$. Since, this supply may not be adequate for serving all the loads, the supplier may have to purchase additional energy at the real-time market price $c^{rt}$. We assume that over the time of actual delivery, the supplier sequentially observes the true realizations of its renewable  power in each slot and makes a real-time purchase decision $a_t$ in slot $t$ according to a \emph{decision policy} of the form
 \begin{equation}
 a_t = g_t(\VEC y, r_1,r_2,\ldots,r_t)
 \end{equation}
 
 For a given day-ahead purchase decision $\VEC y$ and a real-time decision rule $\VEC g = (g_1,g_2,\ldots,g_T)$, the supplier's total expected cost is given as
 \begin{equation} \label{eq:expected_cost}
  \mathcal{J}(\VEC y, \VEC g) := c^{da}\sum_{t=1}^T y_t + c^{rt}\mathds{E} \Bigg[ \sum_{t=1}^T g_t(\VEC y, R_1,R_2,\ldots,R_t)\Bigg],
 \end{equation}
 where the expectation is over the random vector $\VEC R$. The supplier's problem can be stated as follows.
 \begin{problem}
 What choice of the  day-ahead purchase decision $\VEC y$ and the real-time decision rule $\VEC g = (g_1,g_2,\ldots,g_T)$ minimizes the total expected cost given by \eqref{eq:expected_cost} while ensuring that all loads are served?
 \end{problem}

\subsection{Forward market for flexible services}
In Problem 1, the set of duration-differentiated  loads was assumed to be fixed. We would like to investigate  how this set of duration-differentiated loads results from a market interaction between consumers and suppliers. 
We consider a forward market for duration-differentiated services.   All the market transactions between the consumers and the suppliers are completed before the operational period. 
The forward market has three elements:

\balphlist
\item Services: The services are differentiated by the number $h$ of time slots within $T$ slots during which 1 kW of electric power is to be delivered.  Service $h$ is sold at price $\pi(h)$.
\item Consumers: The benefit to a consumer who receives $h$ slots is  $U(h)$; all consumers have the same utility function $U$. 
\item Supplier: The supplier receives for free the power with profile $\VEC R = (R_1, \cdots, R_T)$.    The supplier cost is given by $\mathcal{J}(\VEC y, \VEC g)$ as defined by \eqref{eq:expected_cost} in the previous section.
\end{list}

In order to make the analysis tractable, we focus on idealized market structure  and supplier behavior. In particular, we assume  a competitive market setting in which all agents act as price takers. We also assume that while making their market decisions, suppliers ignore the forecast errors about their renewable supply. 
In this setting, we consider the following problem.
\begin{problem}
 Is there a competitive equilbrium in the forward market for duration-differentiated energy services? Is the equilibrium efficient, that is, does it maximize social welfare?
\end{problem}
\section{Adequacy Results}\label{sec:discrete} 
In order to address Problem 1, we start with providing a characterization of the set of supply profiles that are adequate for a given collection of loads.   We then provide a closed-form expression for the minimum amount of power needed to make an inadequate supply adequate. We use these results to identify the optimal   day-ahead and real-time purchase decisions for the supplier.


\begin{figure*}[ht]
\centering
\begin{tikzpicture}[yscale=0.25, xscale = 0.165]
\draw [xstep=4,ystep=2, thin, dotted] (0,0) grid (24,10);
\node[anchor=west] at (24,1) {$h_1 = 6$};
\node[anchor=west] at (24,3) {$h_2 = 3$};
\node[anchor=west] at (24,5) {$h_3 = 2$};
\node[anchor=west] at (24,7) {$h_4 = 2$};
\node[anchor=west] at (24,9) {$h_5 = 1$};
\node[anchor=north, rotate=90] at (32,5) {Demand profile $h$};
\node[anchor=east, rotate=90] at (2,0) {$d_1 = 5$};
\node[anchor=east, rotate=90] at (6,0) {$d_2 = 4$};
\node[anchor=east, rotate=90] at (10,0) {$d_3 = 2$};
\node[anchor=east, rotate=90] at (14,0) {$d_4 = 1$};
\node[anchor=east, rotate=90] at (18,0) {$d_5 = 1$};
\node[anchor=east, rotate=90] at (22,0) {$d_6 = 1$};
\node[anchor=north] at (12,-5) {Demand duration $d$};
\draw [fill=blue, fill opacity = 0.2, thick] (0,0) rectangle (24,2);
\draw [fill=green, fill opacity = 0.2, thick] (0,2) rectangle (12,4);
\draw [fill=red, fill opacity = 0.2, thick] (0,4) rectangle (8,6);
\draw [fill=yellow, fill opacity = 0.2, thick] (0,6) rectangle (8,8);
\draw [fill=cyan, fill opacity = 0.2, thick] (0,8) rectangle (4,10);

\draw [xstep=4,ystep=2, thin, dotted] (44-4,0) grid (68-4,10);
\draw [fill=blue, fill opacity = 0.2, thick] (44-4,0) rectangle (68-4,2);
\draw [fill=green, fill opacity = 0.2, thick] (48-4,2) rectangle (56-4,4);
\draw [fill=green, fill opacity = 0.2, thick] (60-4,2) rectangle (64-4,4);
\draw [fill=red, fill opacity = 0.2, thick] (48-4,4) rectangle (56-4,6);
\draw [fill=yellow, fill opacity = 0.2, thick] (48-4,6) rectangle (52-4,8);
\draw [fill=cyan, fill opacity = 0.2, thick] (48-4,8) rectangle (52-4,10);
\draw [fill=yellow, fill opacity = 0.2, thick] (64-4,2) rectangle (68-4,4);
\node[anchor=east, rotate=90] at (46-4,0) {$p_1 = 1$};
\node[anchor=east, rotate=90] at (50-4,0) {$p_2 = 5$};
\node[anchor=east, rotate=90] at (54-4,0) {$p_3 = 3$};
\node[anchor=east, rotate=90] at (58-4,0) {$p_4 = 1$};
\node[anchor=east, rotate=90] at (62-4,0) {$p_5 = 2$};
\node[anchor=east, rotate=90] at (66-4,0) {$p_6 = 2$};
\node[anchor=north] at (56-4,-5) {Adequate supply profile $p$};


\draw [xstep=4,ystep=2, thin, dotted] (76-4,0) grid (100-4,10);
\draw [fill=blue, fill opacity = 0.2, thick] (44+28,0) rectangle (64+28,2);
\draw [fill=green, fill opacity = 0.2, thick] (44+28,2) rectangle (56+28,4);
\draw [fill=yellow, fill opacity = 0.2, thick] (56+28,2) rectangle (60+28,4);
\draw [fill=brown, fill opacity = 0.2, thick] (60+28,2) rectangle (64+28,4);
\draw [fill=red, fill opacity = 0.2, thick] (48+28,4) rectangle (56+28,6);
\draw [fill=yellow, fill opacity = 0.2, thick] (48+28,6) rectangle (52+28,8);
\draw [fill=cyan, fill opacity = 0.2, thick] (48+28,8) rectangle (52+28,10);
\node[anchor=east, rotate=90] at (46+28,0) {$p_1 = 2$};
\node[anchor=east, rotate=90] at (50+28,0) {$p_2 = 5$};
\node[anchor=east, rotate=90] at (54+28,0) {$p_3 = 3$};
\node[anchor=east, rotate=90] at (58+28,0) {$p_4 = 2$};
\node[anchor=east, rotate=90] at (62+28,0) {$p_5 = 2$};
\node[anchor=east, rotate=90] at (66+28,0) {$p_6 = 0$};
\node[anchor=north] at (56+28,-5) {Inadequate supply profile $p$};

\end{tikzpicture}
\caption{Characterizing duration-differentiated loads (left), adequate supply profile (center), inadequate supply profile (right).}
\label{fig-demandprofile}
\end{figure*}
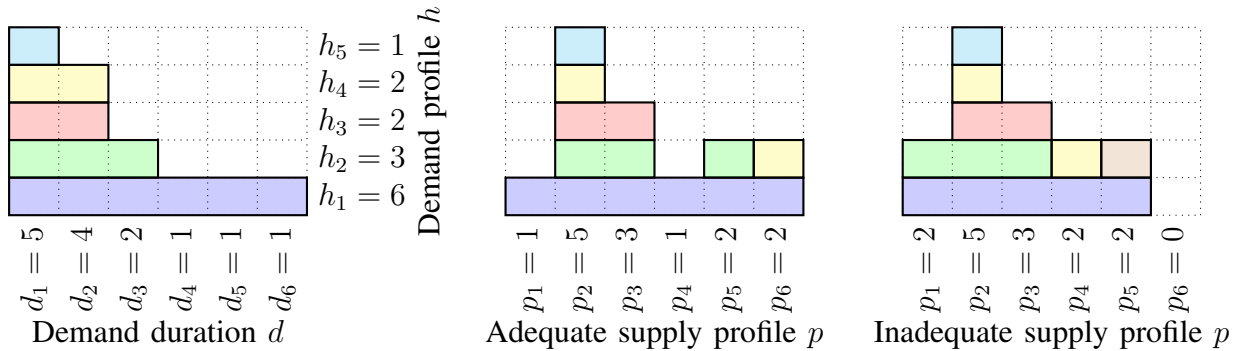

\subsection{Adequacy conditions}
Let $p_t$ be the power available  in time slot $t$. The vector $\VEC p = (p_1,p_2,\ldots,p_T)$ is  the \emph{supply profile}. We identify necessary and sufficient conditions under which $\VEC p$ is adequate for the given collection of loads.  Any allocation of supply to loads can be specified by a binary {\em allocation rule} $A \in \{0,1\}^{N \times T} $ where
$A(i,t) = 1$ if and only if load $i$ receives power in slot $t$.  We  define two notions of supply adequacy.
\begin{definition}[Adequacy] {\rm 
The supply profile $\VEC p=(p_1,\cdots,p_T)$ is \emph{adequate} for the demand profile $\VEC h = (h_1,\cdots,h_N)$ if there exists an allocation rule $A(\cdot,\cdot)$ such that 
\[ \sum_t A(i,t) = h_i, \]
\[ \sum_i A(i,t) \leq p_t.\]
If further,  $\displaystyle{\sum_i A(i,t) = p_t}$, 
we will say that $\VEC p=(p_1,\cdots,p_T)$ is  \emph{exactly adequate} for  $\VEC h = (h_1,\cdots,h_N)$.}
\end{definition}

\begin{example}{\rm Consider $T=6$ time slots and $N=5$ flexible loads as illustrated in Figure \ref{fig-demandprofile}. 
If the demand profile is $\VEC h = (1,2,2,3,6)$, the associated demand-duration vector is $\VEC d = (5,4,2,1,1,1)$.  
The supply profile shown in the center panel is exactly adequate to service the loads. The supply profile shown in the right panel has the same total energy, but it is inadequate to service the loads.}
\end{example}

The following lemma is a direct consequence of the above definition. 
\begin{lemma}
If $\VEC p$ is (exactly) adequate for a demand profile $\VEC h$, then any temporal rearrangement of $\VEC p$ is also (exactly) adequate for the same demand profile $\VEC h$.
\end{lemma}
We will characterize adequacy more directly via the demand-duration vector. 
For this we employ some notions from majorization theory.

\begin{definition}[Majorization]{\rm \label{def:majorization}
Let $\VEC a = (a_1,\cdots,a_T)$ and $\VEC b = (b_1,\cdots,b_T)$ be two non-negative  vectors. Denote by $\VEC a^{\downarrow}, \VEC b^{\downarrow}$ the non-increasing rearrangements of $\VEC a$ and $\VEC b$ respectively. 
We say that $\VEC a$ majorizes $\VEC b$, written $\VEC a \lessthan \VEC b$, if 
\bromalist
\item $\sum_{s=t}^T a^{\downarrow}_s \leq \sum_{s=t}^T b^{\downarrow}_s$, for $t=1,2,\ldots,T$, and 
\item $\sum_{s=1}^T a^{\downarrow}_s= \sum_{s=1}^T b^{\downarrow}_s$.
\end{list}
If only the first condition holds, we say that $\VEC a$ weakly majorizes $\VEC b$, written 
$\VEC a \lessthan^{w} \VEC b$.}
\end{definition}

\begin{remark}{\rm The inequalities in our definition of majorization are reversed from standard use in majorization theory. 
This departure from convention allows us to write our adequacy conditions
as $\VEC d \lessthan \VEC p$ and $\VEC d \lessthan^w \VEC p$  which suggests the intuitive adequacy condition of demand being ``less than'' supply .}
\end{remark}
Our next result characterizes adequacy.

\begin{theorem}[\sc Adequacy]\label{thm:two}
\balphlist
\item The supply profile $\VEC p$  is exactly adequate for a demand profile $\VEC h$ with the associated demand-duration vector $\VEC d$ if and only if $\VEC d \lessthan \VEC p$.
\item The supply profile $\VEC p$  is adequate for a demand profile $\VEC h$ with the associated demand-duration vector $\VEC d$ if and only if 
 $\VEC d \lessthan^{w} \VEC p$.
 \end{list}
\end{theorem}
\proof{Proof}
 See Appendix \ref{sec:two}.
\endproof
\subsection{Least Laxity Allocation}

We now describe an allocation rule that will play a key role in this paper.
Given an allocation rule $A$, we define the laxity of load $i$ at time $t$ as $x_i(t) := T-t+1 - \left(h_i -\sum_{s=1}^{t-1}A(i,s)\right)$. 

\begin{definition}[LLF Allocation]{\rm  \label{defn:LLF}
Fix the supply profile $\VEC p = (p_1,p_2,\ldots,p_T)$. The {\em Least Laxity Allocation} rule $A(i,t)$ is defined by
 \bromalist
  \item At time $1$, $x_i(1) = T-h_i$. Arrange the loads in non-decreasing order of $x_i(1)$ and let $\mathcal{A}_1$ be the collection of the first $p_1$ loads from this order. Set $A(i,1) =1$ if and only if $i \in \mathcal{A}_1$.
  \item At time $t$,  $x_i(t) = T -t+1-\left(h_i -\sum_{s=1}^{t-1}A(i,s)\right)$. Arrange loads in non-decreasing order of $x_i(t)$ and let $\mathcal{A}_t$ be the collection of the first $p_t$ loads from this order. Set $A(i,t) =1$ if and only if $i \in \mathcal{A}_t$.
 \end{list}}
\end{definition}

As its name suggests, at each time $t$ LLF  gives  priority to loads with smaller laxity.
Our next result shows that the LLF allocation successfully services the loads when the supply profile is adequate.  
We have:

\begin{theorem}\label{thm:three}
 If the supply profile $\VEC p = (p_1,p_2,\ldots,p_T)$ is adequate, 
 then the Least Laxity First allocation rule satisfies all the demands, i.e. 
 \[ \sum_t A(i,t) = h_i, \quad \sum_i A(i,t) \leq p_t \]
\end{theorem}
\proof{Proof}
See Appendix \ref{sec:app_3}.
\endproof


\subsection{Supplemental Power Purchases}
It may happen that the supply profile is not adequate for a given demand profile. In this case, the supplier will have to purchase additional power to serve the loads.  We determine the least costly increment in supply profile to make it adequate.
We consider two scenarios: (a)  {\em Oracle information:} the entire supply profile $p$ is revealed in advance, (b)  
{\em Run-time information:} the power available in slot $[t,t+1)$ is revealed at time $t$, i.e. immediately before the beginning of the slot. 

\textbf{Case (a):} \ The supplier needs to serve a demand profile $\VEC h$ with the associated demand-duration vector $\VEC d =(d_1,d_2,\ldots,d_T)$. Before the time of delivery, the supplier learns the true realization of the entire supply profile $\VEC p=(p_1,\ldots,p_T)$. If $\VEC d \lessthan^w \VEC p$,  the supply is adequate and there is no need for supplemental power. In the case of inadequate supply, the additional power to be purchased at minimum cost while ensuring that all demands are met is given by the solution of the following optimization problem:
\begin{equation}
    \min_{\VEC a \geq 0} \sum_{t=1}^T c \cdot a_t \notag \quad
    \mbox{subject to~~~} \VEC d \lessthan^w (\VEC p + \VEC a), \label{LP}
\end{equation}
where $c \geq 0$ is the unit price of supplemental power.  This is a linear programming problem since the majorization inequalities are linear.

\textbf{Case (b):}\  The power available in each slot is revealed just before the beginning of that slot. The supplier now faces  a sequential decision-making problem where the information available to make the purchase decision $a_t$ at time $t$ is $\VEC d, p_1,p_2,\ldots,p_t$. 
The supplier's objective  is to minimize the total cost of  additional power $\sum_{t=1}^Tc\cdot a_t$ while ensuring that all load demands are met.

Clearly the supplier's optimal cost in Case (b) is lower bounded by its optimal cost in Case (a). Surprisingly, it happens that the
optimal costs and corresponding decision strategies are identical in both situations. More precisely, we have:

\begin{figure*}[ht]
\begin{theorem}\label{thm:procurementb}
Consider the following decision strategy for the supplier: \\
(i) The additional power purchased at  $t=1$ is $a_1 = (d_T-p_1)^+$. The total power $(p_1+a_1)$ is allocated to consumers according to the LLF policy described in Definition \ref{defn:LLF}.\\
(ii) At time $t$, knowing the supply $p_1,p_2,\ldots,p_t$ and the purchases $a_1,\ldots,a_{t-1}$, the power purchased $a_t$ is the solution of the following optimization problem:
\begin{align} \nonumber
 &\min_{a_t} a_t  \\
 s.t.~~&(p_1+a_1,p_2+a_2,\ldots,p_{t-1}+a_{t-1},p_t+a_t) \morethan_w (d_{T-t+1},d_{T-t+2},\ldots,d_T)
\end{align}
The total power $(p_t+a_t)$ is allocated to consumers according to the LLF policy.  Then,
\balphlist
\item 
This strategy is optimal under both the oracle information and run-time information cases. 
\item The optimal cost is $c\left( \max_{t} \left(\sum_{ s \geq t} (d_s -p^{\downarrow}_s)\right)^+\right)$.
\end{list}
\end{theorem}

\end{figure*}
\proof{Proof}
See Appendix \ref{sec:procurementb}.
\endproof



\section{The supplier's optimization problem}\label{sec:supplier_problem}
We can now address Problem 1 described in Section~\ref{sec:PF}.  Suppose the supplier purchases a power profile $\VEC y = (y_1,y_2,\ldots,y_T)$ in the day-ahead market. In real-time, the renewable supply takes a realization $\VEC r$, so that the total realized supply profile is $\VEC p = \VEC r+ \VEC y$. Since, this supply may not be adequate, the supplier would have to purchase real-time energy. The optimal policy for making these real-time purchases is given by  the decision strategy described in Theorem~\ref{thm:procurementb} with the associated cost given as $c\left( \max_{t} \left(\sum_{ s \geq t} (d_s -p^{\downarrow}_s)\right)^+\right)$. Hence, the supplier's total expected cost is given as
\begin{align} \label{eq:expected_cost2}
J(\VEC y) = c^{da}\sum_{t=1}^T y_t + c^{rt}\mathds{E}\Bigg[  \max_{t} \left(\sum_{ s \geq t} (d_s -(\VEC R + \VEC y)^{\downarrow}_s)\right)^+\Bigg],
\end{align}
where the expectation is over the renewable supply $\VEC R$. The supplier's optimization problem now is to choose $\VEC y$ to minimize $J(\VEC y)$. Recall that $\VEC y$ is discrete-valued. However, by relaxing the integer constraint on $\VEC y$, we get a convex optimization problem which can be used to provide approximate solutions to supplier's optimization problem.
\begin{theorem}\label{thm:convexity}
If we relax the constraint that $\VEC y$ can take only integer values, then the supplier's objective function is convex in $\VEC y$.
\end{theorem}
\begin{proof}
See Appendix \ref{sec:convexity}.
\end{proof}

\section{Forward market for duration-differentiated contracts}\label{sec:market}
In the analysis presented in the previous sections, the set of duration-differentiated  loads was assumed to be fixed. In this section the problem of arriving to that fixed set is investigated. We consider the case in which the fixed set of duration-differentiated loads is the outcome of a market interaction between consumers and suppliers. 

In particular, we consider a forward market for duration-differentiated services and investigate its properties. Duration-differentiated services couple supply and consumption across different time slots. A natural way to capture this is to consider a forward market for the whole operational period where services of different durations are bought and sold. In this way, both consumers and suppliers can effectively quantify the value/cost of consuming/producing these services.  All the market transactions are completed before the delivery time. Thus, the market decisions are made prior to the operational decisions required for the delivery of the products which is a characteristic of direct load control. In order to illustrate the advantages of forward markets, we perform a comparison with a stylized real-time spot market implementation. The results shed light about potential inefficiencies resulting from the implementation of spot markets in which the inter-temporal dimension of duration-differentiated contracts is hard to capture.

%

In order to make the analysis tractable, we focus on idealized markets structures and consumers and suppliers behaviors. In particular, we assume the following:
\begin{enumerate}
\item Competitive Market: Suppliers and consumers are assumed to be price takers without the possibility of impact the market prices.
\item Certainty Equivalence: In order to make their market decision, the suppliers ignore the uncertainty in their forecast by treating their expected value of renewable power as the true realization. 
\end{enumerate}
The first assumption is standard in market analysis in which price are assumed to be exogenous to the players decisions. The second assumption allows to simplify the suppliers problem in the market setting. It avoids the need to consider the costs resulting from the two-stage optimization problem presenting previously. By assuming this, the additional power cost will be just $c^{da} \sum_{t=1}^T y_t$.

The information flow of the market is depicted in Figure \ref{fig-flow}. Facing a menu of services with associated prices $\mathcal{M}=\{k,\pi(k)\}$, consumer $n$ selects a service $h_n$ that maximizes her net benefit, while the supplier selects the number $n_t$ of services of duration $t$  to produce that  maximize her net profit.  

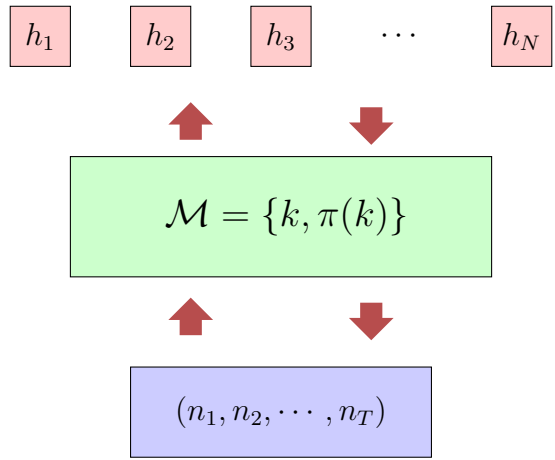
\begin{figure*}[ht]
\centering
\usetikzlibrary{shapes.arrows,fadings}
\begin{tikzpicture}[scale=0.4]
\draw [fill=blue, fill opacity = 0.2] (4,-1) rectangle (14,2);
\draw [fill=green, fill opacity = 0.2] (2,5) rectangle (16,9);
\draw [fill=red, fill opacity = 0.2] (0,12) rectangle (2,14);
\draw [fill=red, fill opacity = 0.2] (4,12) rectangle (6,14);
\draw [fill=red, fill opacity = 0.2] (8,12) rectangle (10,14);
\draw [fill=red, fill opacity = 0.2] (16,12) rectangle (18,14);
\node[] at (1,13) {$h_1$};
\node[] at (5,13) {$h_2$};
\node[] at (9,13) {$h_3$};
\node[] at (13,13) {$\cdots$};
\node[] at (17,13) {$h_N$};
\node[] at (9,7) {\fontsize{14}{16} ${\cal M} = \{k,\pi(k)\}$};
\node[] at (9,0.5) {$(n_1, n_2, \cdots, n_T)$};
\node [
    fill=red!60!black,
    fill opacity=0.7,
    single arrow,
    minimum height=0.5cm,
    single arrow head extend=0.15cm,
    rotate=90
] at (6,3.5){};
\node [
    fill=red!60!black,
    fill opacity=0.7,
    single arrow,
    minimum height=0.5cm,
    single arrow head extend=0.15cm,
    rotate=270
] at (12,3.7){};
\node [
    fill=red!60!black,
    fill opacity=0.7,
    single arrow,
    minimum height=0.5cm,
    single arrow head extend=0.15cm,
    rotate=90
] at (6,10){};
\node [
    fill=red!60!black,
    fill opacity=0.7,
    single arrow,
    minimum height=0.5cm,
    single arrow head extend=0.15cm,
    rotate=270
] at (12,10.2){};
\end{tikzpicture}
\caption{Market implementation.}
\label{fig-flow}
\end{figure*}
We first characterize the decisions that maximize total social welfare, defined as aggregate consumer utility minus the cost of  purchased energy.  We then show that the optimum decisions can be sustained as a  competitive equilibrium. Lastly,   we compare the competitive equilibrium in a real-time spot market with
the equilibrium for our duration-differentiated forward market. 

\subsection{Social Welfare Problem}
We consider a set of homogenous consumers, consumer $i=1,\ldots,N$ enjoys utility $U(h)$ upon consuming 1 kW of power for  $h$ slots. The supplier has available for free a quantity of power with profile $\VEC r = (r_1,\ldots,r_T)$, and can also purchase additional energy at $c^{da}$ per kW-slot. 
 The  social welfare  optimization problem is
\begin{align}
    &\max_{\VEC h \geq 0,\VEC y \geq 0} \sum_{i=1}^N U(h_i) - \sum_{t=1}^T c^{da} y_t \notag \\
    & \mbox{subject to~} 
     d_t = \sum_{i=1}^N \mathds{1}_{\{t\leq h_i\}} \notag \\
    & \hspace{55pt}\VEC y + \VEC r   \morethan^w \VEC d \notag 
\end{align}



\begin{figure*}[ht]

\begin{theorem}[\sc Social Welfare] \label{thm:soc_welfare} 
Assume the profile $\VEC r$ is arranged in decreasing order: $\VEC r = {\VEC r}^{\downarrow}$.
\balphlist
\item (Convex utility)  Suppose $U(h) - U(h-1)$ is a non-negative, non-decreasing function of $h$ (with $U(0)=0$) and the number of consumers $N$ is larger than $r_1$.  Define
\begin{equation}\label{eq:convex_kstar}
k^* = \left\{ \begin{array}{ll}
	\min k : k \in \{0,1,\ldots,(T-1)\}, \ \ \frac{U(T)-U(k)}{T-k} \geq c^{da} & \text{if this exists} \\
	T & \text{otherwise} \end{array} \right.
\end{equation}
	
If  $k^* \geq 1$, the optimum demand duration  is
\begin{align}
d^*_t =\left \{\begin{array}{ll} r_t &\mbox{ if $t < k^*$}\\
                                      r_{k^*} &\mbox{ if $t \geq k^*$} \\                           
                                        \end{array}
                                       \right. \label{eq:soc_welfare1a}
\end{align}    
If $k^* =0$, the optimum demand duration vector is $d^*_t =N$ for all $t$.

\item (Concave utility)  Suppose $U(h)- U(h-1)$ is a non-negative, non-increasing function of $h$ (with $U(0)=0$) and the number of consumers $N$ is larger than $\sum_t r_t$. Define
\begin{equation}\label{eq:concave_kstar}
k^* = \left\{ \begin{array}{ll}
	\min k : k \in \{1,\ldots,T\}, \ \ {U(k)-U(k-1)} \geq c^{da}  & \text{if this exists} \\
	0 & \text{otherwise} \end{array} \right.
\end{equation}
%
If $k^* \geq 1$, the optimum demand duration   is  
\begin{align}
d^*_t =\left \{\begin{array}{ll} N &\mbox{ if $t \leq k^*$}\\
                                      0 & \mbox{ if $t > k^*$}\\
                                        \end{array}
                                       \right. \label{eq:soc_welfare1b}
\end{align} 
If $k^*=0$, the social welfare maximizing demand duration vector is  
\begin{align}
d^*_t =\left \{\begin{array}{ll} 
                                      \sum_{i=1}^T r_i &\mbox { if $t=1$} \\
                                      0  &\mbox{ if $t > 1$}\\
                                        \end{array}
                                       \right. \label{eq:soc_welfare1c}
\end{align} 
\end{list}                 
\end{theorem}

\end{figure*}

\proof{Proof}
See Appendix \ref{sec:soc_welfareproof}.
\endproof

In the convex case, the utility increments are non-decreasing in $h$ and the optimal allocation   favors longer duration contracts. In the concave case,  the utility increments are non-increasing in $h$ and the optimal  allocation   favors the shortest durations. 

\begin{remark}{\rm 
In standard commodity markets, the usual setting is to consider concave utility functions which reflects the decreasing marginal utility of many goods. In the case of duration-differentiated loads, the concave case could represent situations in which additional hours of consumption does not increase the marginal utility, for example the filtering of a pool beyond the minimum numbers of hours. However, in this case convex utility functions are also of interest. That could represent loads for which interruptions of the consumption is material. Examples include industrial mining processes, power supply for computational applications, air flow in hospitals.}
\end{remark}

\begin{example}{\rm 
For example, take  $\VEC r=(5,4,2,1,1,0)$, the number of consumers $N=14$ and the time period $T=6$. In addition, the utility function for the convex case is such that
\begin{equation}
{U(6)-U(5)} \geq c^{da} .
\end{equation}
and for the concave case
\begin{equation}
{U(1)-U(0)} \geq c^{da} .
\end{equation}
In the convex case the optimal allocation is is $\VEC h=(1,2,2,3,6)$ as in Figure \ref{fig-ex1} (left). Note that an additional unit of supply is utilized and used to create a contract of duration 6 hours. In the concave case only contracts of duration 1 slots are required as in Figure \ref{fig-ex1} (middle). Note that in this case, also an additional unit of supply is utilized to create an additional contract of duration 1.}
\end{example}



\begin{figure*}[ht\begin{figure*}]
\centering
\begin{tikzpicture}[scale=0.25]
\draw [xstep=2,ystep=2, thin, dotted] (0,0) grid  (12,10);
\node[anchor=east] at (0,1) {$h_1 = 6$};
\node[anchor=east] at (0,3) {$h_2 = 3$};
\node[anchor=east] at (0,5) {$h_3 = 2$};
\node[anchor=east] at (0,7) {$h_4 = 2$};
\node[anchor=east] at (0,9) {$h_5 = 1$};
\node[anchor=north] at (6,0) {$d =(5,4,2,1,1,1)$};
\draw [fill=blue, fill opacity = 0.2, thick] (0,0) rectangle (12,2);
\draw [fill=green, fill opacity = 0.2, thick] (0,2) rectangle (6,4);
\draw [fill=red, fill opacity = 0.2, thick] (0,4) rectangle (4,6);
\draw [fill=yellow, fill opacity = 0.2, thick] (0,6) rectangle (4,8);
\draw [fill=cyan, fill opacity = 0.2, thick] (0,8) rectangle (2,10);

\draw [fill=blue, fill opacity = 0.2, thick] (22,0) rectangle (24,2);
\draw [fill=green, fill opacity = 0.2, thick] (22,2) rectangle (24,4);
\draw [fill=red, fill opacity = 0.2, thick] (22,4) rectangle (24,6);
\node[] at (23,8) {$\vdots$};
\node[] at (23,11) {$\vdots$};
\node[] at (20,9) {$h_i = 1$};
\draw [fill=violet, fill opacity = 0.2, thick] (22,12) rectangle (24,14);
\draw [fill=yellow, fill opacity = 0.2, thick] (22,14) rectangle (24,16);
\draw [fill=cyan, fill opacity = 0.2, thick] (22,16) rectangle (24,18);
\node[anchor=north] at (23,0) {$d =(14,0,0,0,0,0)$};

\draw [xstep=2,ystep=2, thin, dotted] (34,0) grid (46,10);
\draw [fill=blue, fill opacity = 0.2, thick] (34,0) rectangle (46,2);
\draw [fill=green, fill opacity = 0.2, thick] (34,2) rectangle (40,4);
\draw [fill=red, fill opacity = 0.2, thick] (34,4) rectangle (36,6);
\draw [fill=yellow, fill opacity = 0.2, thick] (34,6) rectangle (36,8);
\draw [fill=red, fill opacity = 0.2, thick] (38,4) rectangle (40,6);
\draw [fill=yellow, fill opacity = 0.2, thick] (38,6) rectangle (40,8);
\draw [fill=cyan, fill opacity = 0.2, thick] (34,8) rectangle (36,10);
\node[anchor=north] at (40,0) {$q =(5,2,4,1,1,1)$};

\end{tikzpicture}
\caption{Optimal allocation for convex case (left), concave case (middle), number of interruptions in the convex case (right).}
\label{fig-ex1}
\end{figure*}
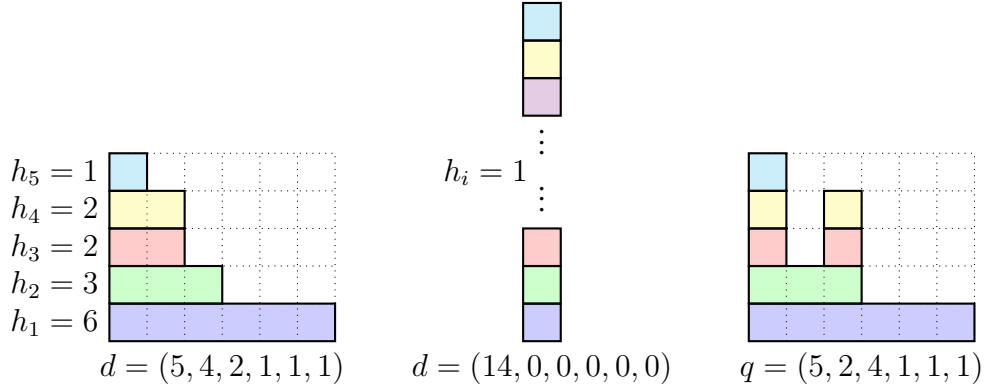

\subsection{Competitive equilibrium}
We now analyze a stylized forward market for the production and consumption of duration-differentiated services. We focus on a perfect-competition setting in which prices are assumed to be \textit{exogenous} to the players' decisions. Consequently, all players are \textit{price takers}, i.e., they cannot influence the prices. The perfect-competition setting is certainly an idealization but it provides valuable insights in terms of market design. In particular, these outcomes are usually used as a benchmark for analysis in which perfect-competition is not considered, e.g., monopolistic or oligopolistic settings. 

The market determines a price for each service, every consumer selects the service she wants to maximize her net benefit, and the supplier decides how much of each service to produce. If the demand and supply for each service match, a competitive equilibrium is said to exist. Mathematically, a competitive equilibrium is defined as follows. 

\begin{definition}[Competitive Equilibrium]{\rm
The supplier's  service production vector $\VEC n = (n_1,n_2,\ldots,n_T)$, the consumers'  demand  profile $\VEC h = (h_1,h_2,\ldots,h_N)$, and a set $\{\pi(h)\}$ of prices  constitute a {\em competitive equilibrium} if  three conditions hold:
\bromalist
\item Consumer Surplus Maximization: $h_i \in \argmax_h U(h) - \pi(h)$, all $i$.
\item Profit Maximization: The services produced maximize profit, i.e. $\VEC n$ solves this optimization problem:
\begin{align}
    &\max_{\VEC n \geq 0,~ \VEC y \geq 0}  \sum_{t=1}^{T}(n_t\pi(t) - Cy_t)  \notag \\
    & \mbox{subject to~} d_t = \sum_{i=t}^T n_i \notag \\
     & \hspace{20pt} \VEC r + \VEC y \morethan^w \VEC d \notag 
\end{align}
\item Market Clearing: The  supply and demand for each service are equal: $n_t = \sum_{i=1}^N \mathds{1}_{\{h_i=t\}}$.
\end{list}}
\end{definition}

\begin{definition}[Efficiency]{\rm
The competitive equilibrium is {\em efficient} if the resulting allocation maximizes social welfare.}
\end{definition}

We next characterize efficient competitive equilibria for convex and concave utility functions.

\begin{theorem}[\sc Efficient Competitive Equilibrium] \label{thm:market} 
Assume the profile $\VEC r$ is arranged in decreasing order: $\VEC r = {\VEC r}^{\downarrow}$.
\balphlist
\item (Convex utility)  Suppose  $U(h)-U(h-1)$ is non-decreasing  (with $U(0)=0$) and the number of consumers $N$ is larger than $r_1$. Then an efficient competitive equilibrium exists and  is described as follows: \\
Prices: $\pi(h)=U(h)$. \\
Supplier's production: Define $k^*$ as in \eqref{eq:convex_kstar}. Then,  if $k^* \geq 1$, $n_t=r_t-r_{t+1}$ for $t<k^*$, $n_t = 0,$ for $k^* \le t <T$, and $n_T=r_{k^*}$ 
If $k^* =0$, the supplier's production is $n_T=N$.\\
Consumption: $n_t$ consumers purchase duration $t$ service, so market is  cleared.\\
\item (Concave utility)  Suppose $U(h)- U(h-1)$ is non-increasing (with $U(0)=0$) and the number of consumers $N$ is larger than $\sum_t r_t$. Then an efficient competitive equilibrium exists and  is described as follows: \\
 Prices: $\pi(h)=\min(c^{da},U(1))h$.\\
Supplier's production:  Define $k^*$ as in \eqref{eq:concave_kstar}. If $k^* \geq 1$, $n_{k^*}=N$ 
If $k^*=0$,  $n_1=\sum_t r_t$ and $n_k=0$ for $k>1$.\\
Consumption:  $n_t$ consumers purchase duration $t$ service, so market is cleared.
\end{list}
\end{theorem}

\proof{Proof} See Appendix \ref{sec:market_proof}. \endproof

\begin{remark}{\rm In essence, Theorem \ref{thm:market} is reminiscent of the first fundamental theorem of welfare economics \cite{mascolell} which states that under certain conditions competitive equilibria are Pareto efficient. Note that our market model differs from the standard market models since each consumer faces a discrete choice \cite{indivisible}.  Further, Theorem \ref{thm:market} is constructive in the sense that it identifies an equilibrium price vector.}\end{remark}


\begin{remark}{\rm In the above analysis, the supply side represents an aggregation of many suppliers. Note that  individual price taking suppliers may benefit by coordinating to offer longer service contracts in order to take advantage of higher prices for longer durations. We assume that the suppliers are able to identify and exploit all such opportunities. We further assume that even with such coordination the number of effective suppliers is large enough to justify the assumption of a competitive market. 
}
\end{remark}

\begin{remark}{\rm In the convex case, the marginal price of the $h$th slot, $\pi(h) - \pi(h-1)$ is \textit{increasing}.  This is due to the fact that from a given 
supply profile it may be possible to produce  a $m$-slot and a $k$-slot service but not a $(m+k)$-slot service, as can be seen from the definition of adequacy.  This contrasts with the assumption in  \cite{Chao-Oren86} that, for conventional generator technology, the marginal price of producing electricity for slot $h$ decreases with $h$. In the concave case above this issue does not arise since service of a single duration is produced.}
\end{remark}

\begin{remark}{\rm 
Suppose the actual chronological profile of the power is $\VEC q = (q_1, \cdots, q_T)$.
Suppose $\VEC q$ has $(m+1)$ local minima.  Then, in the market allocations in the convex case, each service will have at most $m$ interruptions.  For the same example as before Figure \ref{fig-ex1} (right) shows that there is
at most one interruption since $\VEC q = (5,2,4,1,1,1)$ has two local minima.}
\end{remark}

\subsection{Real-time  vs duration-differentiated markets}
We compare the duration-differentiated market with a real-time spot market that operates as follows.  At the beginning of slot $t$ the market determines a price $\pi_t$ for 1kW of power delivered during slot $t$. The price $\pi_t$ equates the supply function $s_t(\pi)$ and demand function $q_t(\pi)$, i.e., $s_t(\pi_t) = q_t(\pi_t)$.  We model these functions.  The supply function is straightforward.  At the beginning of slot $t$, the supplier receives for free power $r_t$ and offers it for sale inelastically.  She can also supply as much power as she wishes at price $C$ per kW-slot.  Thus the supply function is
\begin{equation} \label{a1}
s(\pi) =
\left \{
\begin{array} {ll}
\infty & \mbox{ if } \pi > c^{rt} \\
\mbox{$[p_t,\infty)$} & \mbox{ if } \pi =c^{rt}  \\
p_t & \mbox{ if } \pi < c^{rt}
\end{array}
\right .
.
\end{equation}
The demand function is a bit more complicated.  At the beginning of slot $t$, consumer $n$ determines her willingness to pay $v_n(t)$ for 1 kW in slot $t$. So the (aggregate) demand function is
\begin{equation} \label{a2}
q_t(\pi) = \sum_{i=1}^N \mathds{1}_{\{\pi \le v_n (t)\}}.
\end{equation}
The willingness to pay $v_n(t)$ will depend on the power that consumer $n$   acquired in slots $1, \cdots, t-1$.  All consumers have the same utility function $U(h)$.  If consumer $n$ had  acquired  $x$ kW-slots (out of a total of $t-1$), her willingness to pay would equal  the additional utility she will gain,
\begin{equation} \label{a3}
v_n(t) = U(x+1) - U(x).
\end{equation}
Thus the consumer is myopic: her willingness to pay is unaffected by her opportunities to make future purchases in slots $t+1, \cdots, T$.  The demand function is obtained from \eqref{a2} and \eqref{a3}.  Note that although consumers are myopic, their purchasing decisions depend on  previous decisions. Let $x_i, i=0, \cdots, t-1$, be the number of consumers that have purchased $i$ kW-slots during the first $(t-1)$ spot markets.

Consider the convex case: $U(h)-U(h-1)$ is non-decreasing.  Here, the demand function is

\begin{figure*}[ht]

\begin{equation} \label{a4}
q_t(\pi) = \left \{
\begin{array}{ll}
0 &\mbox{ if } \pi > U(t) - U(t-1) \\
x_{t-1} & \mbox { if }  U(t) - U(t-1) \ge \pi > U(t-1) - U(t-2) \\
x_{t-1}+x_{t-2} & \mbox { if }  U(t-1) - U(t-2) \ge \pi > U(t-2) - U(t-3) \\
& \cdots \\
x_{t-1}+\cdots+ x_1 & \mbox{ if } U(2)-U(1) \ge \pi > U(1) - U(0) \\
x_{t-1} + \cdots +x_0 = N & \mbox{ if } U(1) \ge \pi
\end{array}
\right . .
\end{equation} 
\end{figure*}

Now consider the concave case: $U(h)-U(h-1)$ is non-increasing.  Then the demand function is

\begin{figure*}[ht]
\begin{equation} \label{a5}
q_t(\pi) = \left \{
\begin{array}{ll}
0 &\mbox{ if } \pi > U(1) - U(0) \\
x_0 & \mbox { if }  U(1) - U(0) \ge \pi > U(2) - U(1) \\
x_0+x_1 & \mbox { if }  U(2)-U(1) \ge \pi > U(3) - U(2) \\
& \cdots \\
x_0+\cdots+ x_{t-2} & \mbox{ if } U(t-1) - U(t-2) \ge \pi > U(t-2) - U(t-3) \\
x_0 + \cdots +x_{t-1} = N & \mbox{ if } U(t) -U(t-1)\ge \pi 
\end{array}
\right . .
\end{equation} 
\end{figure*}

\begin{figure*}[ht]
\centering
\includegraphics[width=6in]{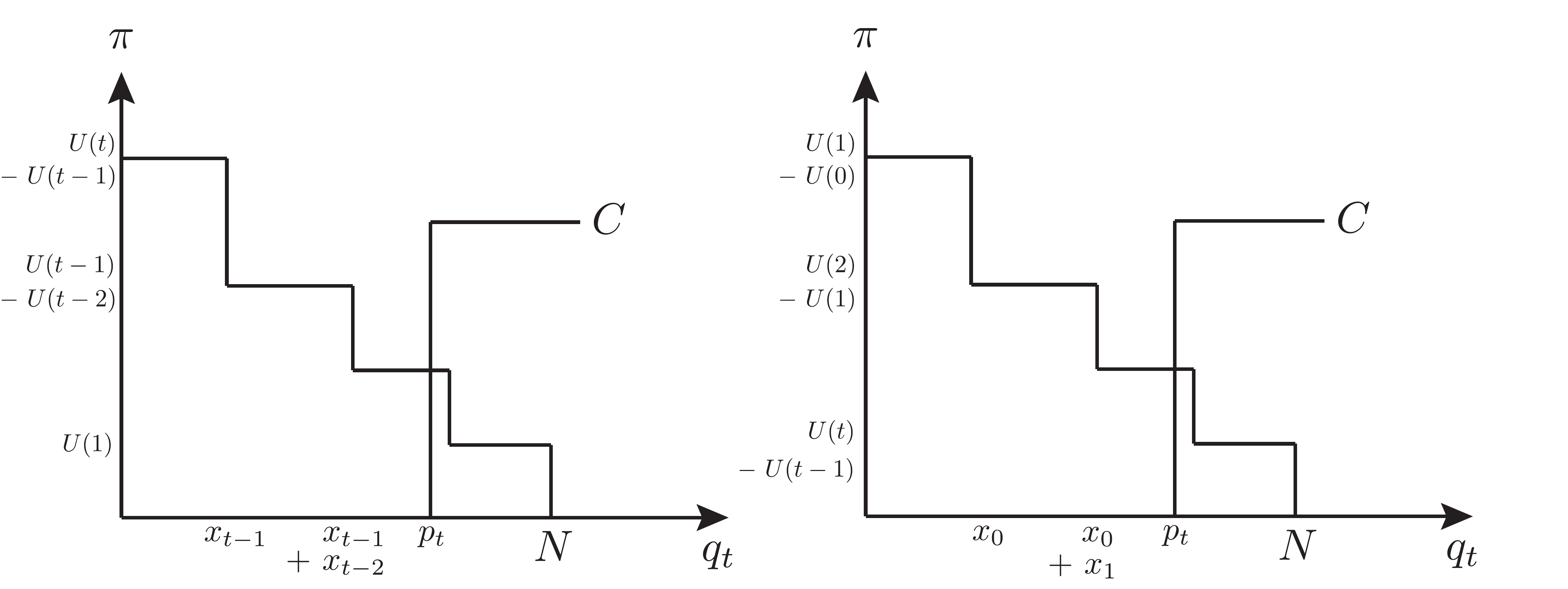}
\caption{Demand and supply functions for spot market: convex utility (left), concave utility (right).}
\label{fig-spot}
\end{figure*}
Figure \ref{fig-spot} shows the supply and demand functions  for the convex case (left) and the concave case (right).

\begin{figure*}[ht]
\begin{theorem}[Real-time market]
\label{thm:rt-market} 
\balphlist
\item Suppose $U(h) - U(h-1)$ is a non-negative, non-decreasing function of $h$ (with $U(0)=0$). 
Then, the spot price at $t$ is
\begin{equation}
\pi_t =\left \{
\begin{array}{ll}  \min(c^{rt},U(t)-U(t-1)) &\mbox{ if $r_t < x_{t-1}$}\\
                                    \min(c^{rt},U(t-1) -U(t-2)) & \mbox{ if $x_{t-1} \leq r_t < x_{t-1} + x_{t-2} $}\\
                                      \min(c^{rt},U(t-2) -U(t-3)) & \mbox{ if $x_{t-1}+x_{t-2} \leq r_t < x_{t-1} + x_{t-2} +x_{t-3}$}\\
                                       \cdots\\
                                       \min(c^{rt},U(1)) & \mbox{ if $x_{t-1}+x_{t-2}+\ldots + x_1 \leq r_t < N$}\\
                                       0 & \mbox{ if $N \leq r_t$}
                                        \end{array}
                                       \right. \label{a6}
\end{equation}

\item Suppose $U(h)- U(h-1)$ is a non-negative, non-increasing function of $h$ (with $U(0)=0$).
Then the spot price at $t$ is
\begin{equation} \label{a7}
\pi_t = \left \{
\begin{array}{ll}
\min(c^{rt}, U(1)-U(0)) & \mbox{ if } r_t < x_0 \\
\min(c^{rt}, U(2)-U(1))& \mbox{ if } x_0 \le r_t< x_0 +x_1  \\
& \cdots \\
\min(c^{rt},U(t)-U(t-1)) & \mbox{ if }  x_0 + \cdots + x_{t-2}   \le        r_t< x_0 + \cdots + x_{t-1} = N \\
0 & \mbox { if } N \leq r_t
\end{array}
\right . .
\end{equation}
\end{list}
\end{theorem}
\end{figure*}
\proof{Proof}
From the intersection of the demand and supply functions in Figure \ref{fig-spot} one can read off the 
formulas \eqref{a6} and \eqref{a7} for the spot price.
\endproof
\begin{example}{\rm   An example shows that the allocation achieved by the real time spot market may not be efficient. In order to compare with the forward market for duration differentiated services, let $c^{da} =c^{rt} = C$. Consider one consumer and two periods, $t=1,2$, the supplier's free power is $r_1 = 0, r_2 = 1$.  For the convex utility case take $U(0) = U(1) = 0, U(2) = 10, C=8$.  In slot 1, since $C > U(1)$, $\pi_1 = U(1)$, the consumer will demand 0 power.
In slot 2, the price will be $\pi_2 = U(1) =0$, the consumer's net benefit will be 0, the producer will receive $\pi_2 = 0$. In the duration-differentiated market  (Theorem \ref{thm:market}), at
the competitive allocation the consumer gets 1 kW for 2 slots and pays $U(2) = 10$ (so her net benefit is 0), the producer will purchase 1 kW in the first slot and pay $C=8$, and will use the free power $p_1$  in the second slot and get a net revenue of $U(2) - 8 = 2$.

For the concave utility case take $U(0)=0, U(1) = U(2) = 5$, $C = 2$.  Since $U(1) > C$, the producer will supply 1 kW in slot 1 at price $\pi_1 = C=2$.
In slot 2, the consumer is indifferent between 1 kW and 0 kW.  She may or not consume $p_2 = 1$ at price $\pi_2 = 0$.
So the consumer's net benefit will be $U(1) - C = 3$, the producer's net profit is 0.  In the duration-differentiated competitive allocation, the consumer will only
purchase 1 kW in slot 2 at price 0, resulting in her net utility of 5, the producer's net profit will again be zero.}
\end{example}

\begin{remark}{\rm  Because of the inter-temporal nature of a consumer's utility, her demand in slot $t$ is contingent on her consumption in other slots.  This contingent demand cannot be met by the system of spot  markets, which is therefore not complete.  This system of markets can be completed through forward markets, like the duration-differentiated market.}
\end{remark}


\section{Rate-constrained energy services} 

We have thus far assumed that consumers demand $1$ kW for a specified number of time slots.
We now consider a more general energy service, which provides a specified quantity of energy to be delivered over $T$ time slots. However, the power level in each time slot can assume any integer value up to a maximum  rate $m$.  
We will argue that these products can be viewed as a combination of the  duration-differentiated services  at a fixed power level of $1$ kW presented before. 

Consumer $i$ specifies two quantities: $(E_i,m_i)$, $E_i$ is the total energy to be consumed over $T$ time slots at a maximum rate of $m_i$ kW per time slot. Both $E_i$ and $m_i$ are integer-valued. 
For example, a consumer specifying $(100,10)$ requires $100$ kW-slots of energy with the constraint that the consumption rate in each time slot is $0,1,2,\ldots$ or $10$ kW. A consumer specifying $(E_i,m_i)$ is satisfied with any supply allocation $A_t \in \{0,1,\ldots,m_i\}$ such that $\sum_t A_t = E_i$. The consumer model of section \ref{sec:discrete} corresponds to the case where $m_i =1$. 

Consider a consumer whose energy requirement and rate constraint are specified by the pair $(E,m)$. Let $k,r$ be such that $E = km + r$ with $r <m$. The following result shows that for the purpose of allocating supply, a consumer specifying $(E,m)$ is equivalent to a collection of $m$ consumers with $(E_n = k+1, m_n =1)$ for $n=1,2,\ldots,r$ and $(E_n = k,m_n=1)$ for $n=r+1,\ldots,m$. Hence, the adequacy and allocation results of Section \ref{sec:discrete} (which were derived for consumers with $m=1$) can be used for the variable rate consumers as well.

\begin{theorem}\label{thm:var_rate}
 Consider a consumer whose energy requirement and rate constraint are specified by the pair $(E,m)$. Let $k,r$ be such that $E = km + r$ with $r <m$. Then, a supply allocation satisfies
 \begin{equation}\label{eq:var_rate1}
  A_t \in \{0,1,\ldots, m\} \quad \sum_t A_t =E
 \end{equation} 
 if and only if 
 \begin{equation} \label{eq:var_rate2}
 A_t = \sum_{n=1}^m A^n_t,
 \end{equation}
 where $A^n_t \in \{0,1\}$ and $\sum_t A^n_t =k+1$ for $n\leq r$ and $\sum_t A^n_t = k$ for $n>r$.
\end{theorem}
\proof{Proof.}
See Appendix \ref{sec:var_rate}.
\endproof

\section{Conclusions and Future Work} \label{sec-conclusions}

Flexible loads play a critical role in enabling deep renewable integration. 
They enable demand shaping to balance supply variability, and thus offer an effective alternative to 
conventional generation reserves. In this paper, we study a stylized model of flexible loads. These loads are modeled
as requiring constant power level for a specified duration within a delivery window. The flexibility resides in the fact that the power delivery may occur at any subset of the total period. 

We consider the supplier's decision problems for serving a given set of duration differentiated loads.  We offered a complete characterization of supply adequacy,
an algorithm for allocating supply to meet the loads and the optimal policies for making day-ahead and real-time power purchase decisions. 
We also study a simplified model of a forward market for these duration differentiated loads where supply uncertainty is not explicitly taken into account.  For this market model, we present the centralized solution that maximizes social welfare and characterize an efficient competitive equilibrium. We also contrast this equilibrium with the outcome of a stylized real-time market.

The theoretical analysis of this work can facilitate further studies required for exploiting load flexibility. For example,  the results of this paper can be used to evaluate the value of flexibility in reducing supplier's cost under different forecast models of renewable generation. Further, since supply adequacy is not assessed at each time instant, forecasting the exact profile of renewable generation may not be necessary.  Instead, the results of this paper suggest that for providing duration-differentiated services, the key metric to be estimated is the supply-duration vector of the renewable generation.

A natural extension of this work is to incorporate supply uncertainty into the forward market for duration-differentiated services.  In particular, the supplier should take into account the expected cost of its  real-time power purchases while making its decisions in the forward market.  In this setting, the supplier's cost function  should be given by $J(\cdot)$ defined in \eqref{eq:expected_cost2}.  The existence and efficiency of  competitive equilibria with this new cost function remains to be ascertained. Further research with heterogeneous consumers that have different utilities and power requirements is needed.

\appendix
\subsection{Preliminary Majorization Based Results}

Let $\VEC a = (a_1,a_2,\ldots,a_T)$ and $\VEC b = (b_1,b_2,\ldots,b_T)$ be two non-negative  vectors  arranged in non-increasing order (that is, $a_t \geq a_{t+1},$ and $b_t \geq b_{t+1}$). The majorization conditions for $\VEC a \lessthan \VEC b$ in Definition \ref{def:majorization} can equivalently be written as 
 \bromalist
 \item $\sum_{s=1}^t b_s \leq \sum_{s=1}^t a_s$, for $t=1,2,\ldots,T-1$, and 
 \item $\sum_{s=1}^T b_s = \sum_{s=1}^T a_s$.
 \end{list}
The first condition is equivalent to the following condition: Let $\VEC c$ be any rearrangement of $\VEC b$; then, for any  $S  \subset \{1,2,\ldots,T\}$, there exists $S' \subset \{1,2,\ldots,T\}$ of the same cardinality as $S$ such that $\sum_{s \in S} c_s \leq \sum_{s \in S'} a_s$.

\begin{definition}{\rm 
 We define a 1 unit {\em Robin Hood (RH) transfer} on $\VEC a$ as an operation that:
 
 \bromalist
 \item Selects indices $t,s$ such that $a_t > a_s$,
 \item Replaces $a_t$ by $a_t-1$ and $a_s$ by $a_s+1$.
 \item Rearranges the resulting vector in a non-increasing order.
 \end{list}}
\end{definition}

\begin{lemma}\label{lemma:1}
 Let $\VEC{\tilde a}$ be a vector obtained from $\VEC a$ after a 1 unit RH transfer. Then, $\VEC a \lessthan \VEC{\tilde a}$.
\end{lemma}
\proof{Proof.} Note $\VEC{\tilde a}$ is a rearrangement of the vector $\VEC{\hat {a}} =(a_1,a_2,\ldots, a_t-1,\ldots, a_s +1,\ldots, a_T)$. 
For any subset $S$ of $\{1,2,\ldots,T\}$,
\bromalist
\item if $t \in S, s\in S$ (or if $t \notin S, s\notin S$), then $\sum_{r \in S} \hat a_r = \sum_{r \in S} a_r$.
\item if $t \in S, s\notin S$, then $\sum_{r \in S} \hat a_r = \sum_{r \in S} a_r -1$.
\item if $t \notin S, s\in S$, then $\sum_{r \in S} \hat a_r = a_s+1 + \sum_{r \in S\setminus\{s\}} \hat a_r \leq a_t +  \sum_{r \in S\setminus\{s\}} \hat a_r = a_t +  \sum_{r \in S\setminus{s}} a_r $. 
\end{list}
Therefore, $\VEC a \lessthan \VEC{\tilde a}$ (see the equivalent condition in the definition of majorization).

\endproof

\begin{lemma}\label{lemma:two}
Suppose $\VEC a \lessthan \VEC b$. If for any $t$, $1 \leq t \leq T$, the following conditions hold:
\bromalist
\item $a_j=b_j$ for all $j < t$,
\item $a_t-a_T \leq 1$.
\end{list}
Then, (a) $a_t=b_t$, and (b) $\VEC a = \VEC b$.
\end{lemma}
\proof{Proof of Lemma~\ref{lemma:two}.}
$\VEC a \lessthan \VEC b$ and $a_j=b_j$ for all $j < t$ imply that 
\begin{align}
&\sum_{i=t}^s b_i \leq \sum_{i=t}^s a_i, \quad s=t,t+1,\ldots,T-1,  \notag \\
&\sum_{i=t}^T a_i = \sum_{i=t}^T b_i
\end{align}
In particular, $b_t \leq a_t$.
If $b_t < a_t$, then $\sum_{i=t}^T b_i \leq (T-t+1)(a_t-1)$. On the other hand, since $a_T \geq a_t-1$. $\sum_{i=t}^T a_i > (T-t+1)(a_t-1)$. Therefore, $\sum_{i=t}^T a_i \neq \sum_{i=t}^T b_i$, which contradicts the conditions of the lemma. Thus, $b_t$ must be equal to $a_t$. Reapplying the first part of the lemma for $i=t+1$, then gives $a_{t+1}=b_{t+1}$. Proceeding sequentially till $T$ proves the second part of the lemma. 
\endproof

\begin{lemma} \label{lemma:3}
Let $\VEC a \lessthan \VEC b$ and $\VEC a \neq \VEC b$. Then, there exists a 1 unit RH operation on $\VEC a$ that gives a vector $\VEC{\tilde a } \neq \VEC a$ satisfying $\VEC a \lessthan \VEC{\tilde a} \lessthan \VEC b$.
\end{lemma}
\proof{Proof.} 
Let $t$ be the smallest index  such that $a_t \neq b_t$. Then, $b_t<a_t$ since $\sum_{i=1}^t b_i \leq \sum_{i=1}^t a_i$ and the two vectors have the same first $t-1$ elements. Let $s>t$ be the smallest index such that $a_t-a_s>1$.  Such $s$ must exist, otherwise Lemma \ref{lemma:two} would imply that $\VEC a = \VEC b$. Consider a 1 unit RH transfer from $t$ to $s$. Let $\VEC{\tilde a}$ be the resulting vector. Then, by Lemma \ref{lemma:1}, $\VEC a \lessthan \VEC{\tilde a}$. 

Also, if $k$ is the number of elements of $\VEC a$ with value equal to  $a_t$, then the number of elements of $\VEC{\tilde a}$ with value equal to $a_t$ is $k-1$. Therefore, $\VEC{\tilde a } \neq \VEC a$.
 
Further, it is clear that $\VEC a$ and $\VEC {\tilde{a}}$ have the same first $t-1$ elements (since the RH operation depleted 1 unit from $a_t$ and added it to $a_s<a_t-1$, the non-increasing rearrangement would not change the $t-1$ highest elements.) Similarly, $\VEC a$ and $\VEC{\tilde{a}}$ have the same elements from index $s$ to $T$. Further, from the definition of $s$, $a_j \geq a_t -1$ for $t \leq j <s$. Since,  $\tilde{a}_t,\tilde{a}_{t+1} \ldots, \tilde{a}_{s-1}$, must be a rearrangement of $a_{t}-1,a_{t+1}\ldots, a_{s-1}$, it follows that $\tilde a_j \geq a_t -1$ for $t \leq j <s$.

We now prove that $\VEC{\tilde a} \lessthan \VEC b$. 
\bromalist
\item If $j < t$ or $j \geq s$, then $\sum_{i=1}^j \tilde a_i = \sum_{i=1}^j a_i \geq \sum_{i=1}^j b_i$.
\item If $t \leq j < s$, then 
\begin{align}
&\sum_{i=1}^j \tilde a_i = \sum_{i=1}^{t-1}  a_i + \sum_{i=t}^{j}  \tilde a_i  \notag\\
& =\sum_{i=1}^{t-1}  b_i +\sum_{i=t}^{j}  \tilde a_i \geq \sum_{i=1}^{t-1}  b_i +(j-t+1)(a_t-1), \label{eq:one}
\end{align}
where the first and second equalities are true because the first $t-1$ elements of $\VEC a, \VEC{\tilde a}$ and $\VEC b$ are the same, and  last inequality follows from the fact that for $t \leq i <s$, $\tilde{a}_i \geq (a_{t}-1)$.
Moreover, 
\begin{align}
\sum_{i=1}^jb_i = \sum_{i=1}^{t-1}  b_i+  \sum_{i=t}^jb_i \leq \sum_{i=1}^{t-1}  b_i+ (j-t+1)b_t \leq \sum_{i=1}^{t-1}  b_i+ (j-t+1)(a_t-1),\label{eq:two}
\end{align}
where the last inequality follows from $b_t < a_t$. Equations \eqref{eq:one} and \eqref{eq:two} imply that
$\sum_{i=1}^jb_i \leq \sum_{i=1}^j \tilde a_i$, or $\VEC{\tilde a} \lessthan \VEC b$. 
\end{list}
\endproof

\begin{claim}
 Let  $\VEC a \lessthan \VEC b$ with $\VEC a \neq \VEC b$. Then, there exists a finite sequence of 1 unit RH transfers that can be applied on $\VEC a$ to get $\VEC b$.
\end{claim}
\proof{Proof. } The claim is established using Lemmas \ref{lemma:two} and \ref{lemma:3}.
Let $\VEC a^0 = \VEC a$.
\bromalist
\item For $n =1,2,\ldots$, if $\VEC a^{n-1} \neq \VEC b$,  use Lemma \ref{lemma:3} to construct $\VEC a^n \neq \VEC a^{n-1}$ such that $\VEC a^{n-1} \lessthan \VEC a^n \lessthan \VEC b$. Then, $\VEC a^n \neq \VEC a^m$ for any $m<n-1$ 
(otherwise, we would have $\VEC a^m=\VEC a^n \lessthan \VEC a^{n-1} \lessthan \VEC a^n \implies \VEC a^n=\VEC a^{n-1}$).
\item If  $\VEC a^{n-1} = \VEC b$, stop.
\end{list}

Since there are only finitely many non-negative integer valued vectors that majorize $\VEC b$, this procedure must eventually stop and it can do so only if $\VEC a^{n-1}=\VEC b$, proving the claim.                                 
\endproof

\subsection{Proof of Theorem \ref{thm:two}} \label{sec:two}

We require with the following intermediate result.
\begin{lemma}\label{lemma:one}
For a given demand profile, if $\VEC a$ is an exactly adequate supply profile, then any  supply profile $\VEC b$ satisfying $\VEC a \lessthan \VEC b$  is also exactly adequate.
\end{lemma}
\proof{Proof.} Without loss of generality, we will assume that $\VEC a$ and $\VEC b$ are arranged in non-increasing order.
Since $\VEC b$ can be obtained from $\VEC a$ by a sequence of 1 unit RH transfers (see Claim 1, Appendix A), we simply need to prove that a 1 unit RH transfer preserves exact adequacy. Consider an exactly adequate profile $\VEC a$ and let $A$ be the corresponding allocation function. Consider  a 1 unit RH transfer  from time $t$ to $s$ (without rearrangement) that gives a new profile $\VEC{\tilde a}$.   Let $i$ be a load for which $A(i,t) =1$ but $A(i,s) =0$. Such an $i$ must exist because $\sum_j A(j,t) = a_t > a_s = \sum_j A(j,s)$. Under the new profile, the allocation rule 
\begin{align}
\tilde{A}(j,r) =\left \{\begin{array}{ll} & A(j,r) ~~r\neq t,s \mbox{~or~} j\neq i\\
                                      &0 ~~r=t,j=i \\
                                      &1 ~~r=s,j=i \\
                                        \end{array}
                                       \right.
\end{align} 
establishes exact adequacy.                      
\endproof

\proof{\bf Proof of Theorem \ref{thm:two} (a):}  \  Observe that $\VEC d $ is exactly adequate: the allocation function $A(i,t) = \mathds{1}_{\{t \leq h_i\}}$ meets the exact adequacy requirements under $\VEC d$. Therefore, any $\VEC p$ satisfying $\VEC p \morethan \VEC d$ is also exactly adequate. 
 
To prove necessity, suppose $\VEC p$ is exactly adequate and $A$ is the corresponding allocation function. Consider any set $\mathcal{S} \subset \{1,2,\ldots,T\}$ of cardinality $s$. Then, because $A(i,t) \in \{0,1\}$ and $\sum_{t=1}^T A(i,t) = h_i$, it follows that
\begin{equation} \label{eq:necc1}  \sum_{t \in \mathcal{S}} A(i,t) \leq \min(s,h_i) = \sum_{t=1}^s \mathds{1}_{\{t \leq h_i\}}. \end{equation}
Summing over $i$,
\begin{align}
 &~~\sum_{i=1}^N\sum_{t \in \mathcal{S}} A(i,t) \leq \sum_{i=1}^N\sum_{t =1}^s \mathds{1}_{\{t \leq h_i\}} \notag \\
 &\implies \sum_{t \in \mathcal{S}}\sum_{i=1}^N A(i,t) \leq \sum_{t=1}^s \sum_{i=1}^N\mathds{1}_{\{t \leq h_i\}} \notag \\
 &\implies \sum_{t \in \mathcal{S}} p_t \leq \sum_{t=1}^s d_t  \quad \mbox{for all $\mathcal{S}$ with $|\mathcal{S}| =s$}\notag \\
 &\implies   \sum_{t=1}^s p^{\downarrow}_t\leq \sum_{t=1}^s d_t\label{eq:necc2}
\end{align}
For $s=T$, the inequality in \eqref{eq:necc1} becomes an equality resulting in an equality in \eqref{eq:necc2}. 
\endproof

\proof{\bf Proof of Theorem \ref{thm:two} (b):}
To prove necessity, suppose $\VEC p$ is adequate and $A$ is the corresponding allocation function. Then, clearly, for any set $\mathcal{S} \subset \{1,\ldots,T\}$ of cardinality $s$
\begin{equation}
\sum_{t \in \mathcal{S}} p_t \geq \sum_{t \in \mathcal{S}} \sum_{i=1}^N A(i,t)\label{eq:app4_eq1}
\end{equation}
 Further, because $A(i,t) \in \{0,1\}$ and $\sum_{t=1}^T A(i,t) = h_i$, it follows that
\begin{align}  
 &\sum_{t \in \mathcal{S}} A(i,t) \geq \max(h_i-s+1,0) = \sum_{t=s}^T \mathds{1}_{\{t \leq h_i\}} \notag  
 \end{align}
Summing over $i$,
\begin{align}
 &~~\sum_{i=1}^N\sum_{t \in \mathcal{S}} A(i,t) \geq \sum_{i=1}^N\sum_{t=s}^T \mathds{1}_{\{t \leq h_i\}} \notag \\
 &\implies \sum_{t \in \mathcal{S}}\sum_{i=1}^N A(i,t) \geq \sum_{t=s}^T \sum_{i=1}^N\mathds{1}_{\{t \leq h_i\}} \notag \\
 &\implies \sum_{t \in \mathcal{S}}\sum_{i=1}^N A(i,t) \geq  \sum_{t=s}^T d_t. \label{eq:app4_eq2}
\end{align}
Combining \eqref{eq:app4_eq1} and \eqref{eq:app4_eq2} proves the necessity conditions.

To prove sufficiency, let $\Delta = \sum_{t=1}^T p_t - \sum_{t=1}^T d_t$. Then, $\Delta \geq 0$. Consider a new demand profile $\VEC d^{\Delta}$ defined as $d^{\Delta}_t := d_t + \Delta\mathds{1}_{\{t \leq 1\}}$. This new demand profile corresponds to the original demand profile augmented  with $\Delta$ ``fictitious loads'' each requiring $1$ unit of energy for $1$ time slot. It is easy to see that $\VEC p \morethan^w \VEC q$ implies $\VEC p \morethan \VEC d^{\Delta}$. Therefore, by Theorem \ref{thm:two}, $\VEC p$ is exactly adequate for the augmented demand profile $\VEC d^{\Delta}$ which implies that it must be adequate for the original demand profile $\VEC q$. 
\endproof

\subsection{Proof of Theorem \ref{thm:three}}\label{sec:app_3}
Since $\VEC p$ is adequate, there must exist at least one allocation function $B(i,t)$ such that
\[ \sum_t B(i,t) = h_i, \quad  \sum_i B(i,t) \leq p_t.\]
Let $\mathcal{B}_1$ be the set of loads served at time $1$ under the allocation rule $B(i,t)$. If $\mathcal{A}_1 \neq \mathcal{B}_1$, pick a load $i \in \mathcal{A}_1\setminus\mathcal{B}_1$ and $j \in \mathcal{B}_1\setminus\mathcal{A}_1$. Then, $h_i \geq h_j$. Therefore, there must exist a time $s>1$ such that $B(i,s) =1$ but $B(j,s)=0$. Consider a new allocation rule $B^1(\cdot,\cdot)$ obtained by swapping the load $i$  at time $s$ and the load $j$ at time $1$, that is, the new allocation rule $B^1$ is identical to $B$ except that for $t=1,s$
\[ B^1(i,t) = B(j,t) \quad \mbox{and} \quad B^1(j,t)=B(i,t). \]
It is straightforward to establish that the new allocation rule satisfies the adequacy requirements. One can proceed  with this swapping argument one load at a time until the set of loads being served at time $1$ is equal to $\mathcal{A}_1$.

For time $t>1$, the same swapping argument works after $h_i$ are replaced by the number of time slots each load needs to be served from time $t$ to the final time, which is precisely $x_i(t)$. 

\subsection{Proof of Theorem \ref{thm:procurementb}}\label{sec:procurementb}

We first observe that the strategy described in the theorem is a valid strategy under Case B since it only uses the information available until time $t$. Clearly, the prescribed strategy is valid under Case A as well. We will now show that it yields an adequate supply in the end.

From the definition of $a_1$, we have $(p_1+a_1) \morethan^w d_T$. Now assume that $(p_1+a_1,\ldots,p_{t-1}+a_{t-1}) \morethan^w (d_{T-t+2},\ldots,d_T)$. Then, there exists an $a_t$ such that 
$(p_1+a_1,p_2+a_2,\ldots,p_{t-1}+a_{t-1},p_t+a_t) \morethan^w (d_{T-t+1},d_{T-t+2},\ldots,d_T)$. 
Thus, the supplier's optimization problem at time $t$ has a solution. This is given by $\min a_t $ subject to
\begin{eqnarray*}
& & p_t+a_t \geq d_T \\
& & p_s+a_s+p_t+a_t \geq d_T + d_{T-1} \quad \mbox{for all~}s <T \\
& & \sum_{s=s_1,s_2,\ldots,s_k}(p_s+a_s) + (p_t+a_t) \geq d_{T-k}+\cdots+ d_T \quad \mbox{for all~} s_1 < \cdots < s_k <t  \mbox{ and~} k \leq t-1
\end{eqnarray*}
It is clear that the prescribed decision strategy results in an adequate net supply since at time $T$, the supplier's optimization ensures that $\VEC p+\VEC a \morethan^w \VEC d$.

We now argue that the prescribed strategy is the optimal one. Let $\VEC b = (b_1,\ldots,b_T) \geq 0$ be the optimal sequence of purchase decisions made in Case A Then, adequacy constrain implies that $\VEC p + \VEC b \morethan^w \VEC d$. Because, the total supply is now adequate, it can be allocated to consumers according to the LLDF rule.

 Starting at time $1$, we must have $b_1 \geq  a_1 =(d_T-p_1)^+$. Otherwise, $(p_1 + b_1) < d_T$ which means that $\VEC p + \VEC b$ could not be adequate for the given demand.  Suppose $b_1 > a_1$. Let $\mathcal{B}_1$ be the set of consumers served under a LLDF allocation rule at time $1$ when the supply is $\VEC p + \VEC b$ and $\mathcal{A}_1$ be the set of consumers served under a LLDF allocation rule at time $1$ when the supply is $\VEC p + \VEC a$. Since, $p_1 + b_1 \geq p_1 + a_1$, $\mathcal{B}_1 \supset \mathcal{A}_1$. Let $i \in \mathcal{B}_1 \setminus \mathcal{A}_1$. Because, we know that $\VEC p + \VEC a$ is adequate and that consumer $i$ is not being served at time $1$, there must be a future time $s$ such that consumer $i$ is served at $s$ under the supply $\VEC p + \VEC a$ but not under $\VEC p + \VEC b$. We now consider a modification of the purchase decisions $\VEC b$ constructed as follows: 
\begin{enumerate}
\item Change $b_1$ to $b_1-1$ and not serve consumer $i$ at time $1$.
\item At time $s$, if there was excess power available when the supply was $\VEC p + \VEC b$, then use it to serve consumer $i$ at time $s$. Otherwise, change $b_s$ to $b_s +1$ with the new unit of power given to consumer $i$.
\end{enumerate}
Denote this modified vector of decisions by $\VEC{\tilde{b}}$. It is clear that this modified decision is also optimal since it ensures adequacy and its total cost does not exceed the cost of $\VEC b$. We can repeat this argument until $\tilde{b}_1 = a_1$. Therefore,  there is an optimal vector of purchase decisions of the form
$(a_1,\tilde{b}_2,\ldots,\tilde{b}_T)$.

We now proceed by induction. Suppose that the optimal decision vector is $(a_1,\ldots,a_{t-1},b_t,\ldots,b_T)$. Then, $b_t \geq a_t$ otherwise $\VEC p + \VEC b \not \morethan^w d$. If $b_t >a_t$, we can use the same rearrangement argument used at time $1$ to construct modifications of the optimal decision vector that is of the form $(a_1,\ldots,a_{t},\tilde{b}_{t+1},\ldots,\tilde{b}_T)$. Continuing sequentially till the final time $T$, we conclude that $(a_1,\ldots,a_T)$ must be optimal.

For any vector $\VEC x = (x_1,\ldots,x_T)$ define $F(\VEC x) := \max_{\mathcal{S} \subset \{1,2,\ldots,T\}} \left(\sum_{ s= T- |\mathcal{S}|+1}^T d_s - \sum_{r \in \mathcal{S}} x_r\right)^+$. 

 If $a_1 >0$, then $d_T >p_1$ and therefore the maximum achieving set $\mathcal{S}$ in definition of $F(\VEC p)$ must contain time index $1$. Let $\VEC p(1) = (p_1+a_1,p_2,\ldots,p_T)$. Then, $\mathcal{S}$ also achieves the maximum in the definition of $F(\VEC p(1))$. Therefore,
\begin{equation}
F(\VEC p(1)) = F(\VEC p) -a_1.
\end{equation}

Let $\VEC p(t-1)$ be the supply profile after the first $t-1$ purchase decisions. If $a_t >0$, then there exists a set $\mathcal{T} \subset \{1,\ldots,t\}$ containing $t$ such that 
\begin{equation}
  \sum_{s=T-|\mathcal{T}|+1}^T d_s - \sum_{r \in \mathcal{T}} p_r(t-1) >0
\end{equation}
Let $\mathcal{T}^*$ be the set for which the above difference is the largest. Then, the maximum achieving set $\mathcal{S}$ in the definition of $F(\VEC p(t-1))$ must contain $\mathcal{T}^*$. Define $\VEC p(t) = (p_1+a_1,p_2+a_2,\ldots,p_t+a_t,p_{t+1}\ldots,p_T)$. Then, $\mathcal{S}$ also achieves the maximum in the definition of $F(\VEC p(t))$. Therefore,
\begin{equation}
F(\VEC p(t)) = F(\VEC p(t-1)) -a_t.\label{eq:costbound2}
\end{equation}
Combining \eqref{eq:costbound2} for all $t$ gives
\begin{equation}
F(\VEC p(T)) = F(\VEC p) - \sum_{t=1}^T a_t \notag \quad
\implies \quad  \sum_{t=1}^T a_t = F(\VEC p) - F(\VEC p(T)). \label{eq:costbound3}
\end{equation}
Note that 
\begin{align}
F(\VEC p) =\max_{\mathcal{S} \subset \{1,2,\ldots,T\}} \left(\sum_{ s= T- |\mathcal{S}|+1}^T d_s - \sum_{r \in \mathcal{S}} p_r\right)^+ \notag \\
= \max_{1 \leq t \leq T}\left(\sum_{ s= T-t+1}^T d_s - \sum_{ i= T-t+1}^T p^{\downarrow}_r\right)^+
\end{align}
Also note that since $\VEC p(T) = \VEC p+\VEC a$, and $\VEC p + \VEC a \morethan^w \VEC d$, it implies that 
\[F(\VEC p(T)) = \max_{1 \leq t \leq T}\left(\sum_{ s= T-t+1}^T d_s - \sum_{ r= T-t+1}^T p^{\downarrow}_r(T) \right)^+ =0.\] 
Therefore, \eqref{eq:costbound3} amounts to
\[\sum_{t=1}^T a_t  =\max_{1 \leq t \leq T}\left(\sum_{ s= T-t+1}^T d_s - \sum_{ r= T-t+1}^T p^{\downarrow}_r \right)^+ = \max_{1 \leq t \leq T}\left(\sum_{ s=t}^T d_s - \sum_{ r= t}^T p^{\downarrow}_r \right)^+ \]
which proves the theorem. 

\subsection{Proof of Theorem \ref{thm:convexity}} \label{sec:convexity}
In order to prove the convexity of $J(\VEC y)$, it suffices to prove the convexity of 
\[ \max_{t} \left(\sum_{ s \geq t} (d_s -(\VEC r +\VEC y)^{\downarrow}_s)\right)^+\]
for each possible realization $\VEC r$ of the renewable supply. Since maximum of convex functions is convex, it is sufficient to prove that 
\begin{equation} \label{eq:convexity1} \left(\sum_{ s \geq t} (d_s -(\VEC r +\VEC y)^{\downarrow}_s)\right)^+
\end{equation}
is convex for all $t$. \eqref{eq:convexity1} can be written as
\begin{align}
&= \sum_{ s =t}^{T} d_s - \min_{\substack{\mathcal{S} \subset \{1,2,\ldots,T\},\\ |\mathcal{S}| = T-t+1}} \left( \sum_{k \in \mathcal{S}} (r_k +y_k)\right) \notag \\
&=\sum_{ s =t}^{T} d_s + \max_{\substack{\mathcal{S} \subset \{1,2,\ldots,T\},\\ |\mathcal{S}| = T-t+1 }} \left( -\sum_{k \in \mathcal{S}} (r_k +y_k)\right) \label{eq:convexity2}
\end{align}
Since \eqref{eq:convexity2} is a maximum of affine functions of $\VEC y$, it implies that it is convex in $\VEC y$. This completes the proof.

\subsection{Proof of Theorem \ref{thm:soc_welfare}}\label{sec:soc_welfareproof}
\proof{\bf Proof of Part (a)}
Consider any choice of $\VEC y$ so that the total supply is $\VEC x = \VEC r + \VEC y$. Without loss of generality, assume that $\VEC x$ is arranged in non-increasing order.
For $t=1,\ldots,T$, define $\delta_t=U(t)-U(t-1)$ and define $\delta_0=0$. Then, $\delta_0 \leq \delta_1 \leq \delta_2 \leq \ldots \leq \delta_T$.
For any demand duration vector that can be served with $\VEC x$, the total utility of consumers is 
\begin{eqnarray}
 \sum_{t=1}^{T-1}(d_t-d_{t+1}) U(t) + d_TU(T) & = &  \sum_{t=1}^T d_t(U(t)-U(t-1)) = \sum_{t=1}^T d_t\delta_t  \nonumber \\
 & = & \delta_1(d_1+d_2+\ldots+d_T) + (\delta_2-\delta_1)(d_2+d_3+\ldots+d_T)  \nonumber \\
 & & \quad \quad \quad +  (\delta_3-\delta_2)(d_3+\ldots+d_T) + \ldots + (\delta_T-\delta_{T-1})d_T \nonumber \\
 & = &   \sum_{j=1}^{T}(\delta_j-\delta_{j-1})\left[\sum_{t=j}^T d_t\right]\label{eq:soc_welfare1}
\end{eqnarray}
where the inequality follows from the fact that $\delta_j -\delta_{j-1} \geq 0$ and $\VEC x \morethan^w \VEC d$. The profile $d^*_t = x_t$ for all $t$ is a valid demand profile that achieves the upper bound on consumer utility. Therefore, the optimal contracts offered when the supply is $\VEC x$ are: the first $x_T$ consumers get $T$ duration, the next $x_{T-1} -x_{T}$ get $T-1$ and in general $x_{i}-x_{i+1}$ get duration $i$. In total, $x_1$ consumers are served. If $\VEC x = \VEC r$, the first $r_T$ consumers get $T$ duration, the next $r_{T-1} -r_{T}$ get $T-1$ and in general $r_{i}-r_{i+1}$ get duration $i$. Since $\VEC x = \VEC r+\VEC y$ and $\VEC y \geq \VEC 0$, each consumer's contract under $\VEC x$ is no less than its contract under $\VEC r$. Therefore, the social welfare problem can be restated as follows: Let $(h_1,h_2,\ldots,h_N)$ be the utility maximizing contracts for the $N$ consumers under the supply $\VEC r$. Choose a vector $\VEC z = (z_1,z_2,\ldots,z_N)$ of contract increments to maximize
\begin{equation}
\sum_{i=1}^N \left(U(h_i+z_i) - c^{da}z_i\right)
\end{equation}
The summation can be maximized by maximizing each term separately which gives that $z_i = T -h_i$ if $h_i \geq k^*$ and $0$ otherwise. 
\endproof
  
\proof{\bf Proof of part (b)}
Consider any choice of $\VEC y$ so that the total supply is $\VEC x = \VEC r + \VEC y$. Let $\sum_{t=1}^T x_t = Nu(x) + S(x)$ where $u(x)$ is a non-negative integer and $S(x) < N$. Then, the non-increasing increment property of the utility implies that total consumer utility is maximized if $S(x)$ consumers get contracts of duration $u(x)+1$ and $N-S(x)$ consumers get contracts of duration $u(x)$. When $\VEC x = \VEC r$, $\sum_t r_t$ consumers get contract of duration $1$. Thus, each consumer's contract under $\VEC x$ is no less than its contract under $\VEC r$. Therefore, the social welfare problem can be restated as follows: Let $(h_1,h_2,\ldots,h_N)$ be the utility maximizing contracts for the $N$ consumers under the supply $\VEC r$. Choose a vector $\VEC z = (z_1,z_2,\ldots,z_N)$ of contract increments to maximize
\begin{equation}
\sum_{i=1}^N \left(U(h_i+z_i) - c^{da}z_i\right)
\end{equation}
The summation can be maximized term-by-term which yields $z_i = k^* -h_i$. 
\endproof

\subsection{Proof of Theorem \ref{thm:market}} \label{sec:market_proof}

\proof{\bf Proof of Part (a)}
Subject to prices $\pi(t)=U(t)$, the profit maximization problem is given by
\begin{align}
    &\max_{\VEC n \geq 0,~ \VEC y \geq 0}  \sum_{t=1}^{T}(n_tU(t) - c^{da}y_t)  \notag \\
    & \mbox{subject to~} d_t = \sum_{i=t}^T n_i \notag \\
     & \hspace{20pt} \VEC r + \VEC y \morethan^w \VEC d \notag 
\end{align}
but  $n_i=\sum_{j=1}^N \mathds{1}_{\{h_j=i\}}$, then $\sum_{t=1}^{T}n_tU(t)$ is equal to $\sum_{i=1}^N U(h_i)$. Thus, the solution of the profit maximization problem is equivalent to the solution of the social welfare optimization problem. Subject to these prices, consumers obtain zero welfare under any allocation. Hence, an efficient competitive equilibrium exists. 
\endproof

\proof{\bf Proof of part (b)}
Prices are given by $\pi(h)=\min(c^{da},U(1))h=\mu h$. The case in which $\pi(h)=U(1)h$ is straightforward. If  $\pi(h)=c^{da}h$, it is clear that consumers maximize their surplus by choosing contracts of duration $k^*$. This is a result of the non-increasing increments on $U(h)$ and the condition $U(k^*)-U(k^*-1) \geq c^{da}$. Hence, the bundle of contracts $n_k*=N$ maximizes consumers surplus. Subject to prices $\pi(h)=c^{da}h$, supplier revenue is given by $c^{da} \sum_{t=1}^T n_t t - c^{da}\sum_{t=1}^T y_t=c^{da} \sum_{t=1}^T d_t - c^{da}\sum_{t=1}^T y_t=c^{da}(\sum_{t=1}^T d_t -\sum_{t=1}^T y_t) \leq c^{da} \sum_{t=1}^T r_t$. The bundle, $n_k*=N$ achieves the upper bound. Hence, an efficient competitive equilibrium exists.  
\endproof

\subsection{Proof of Theorem \ref{thm:var_rate}} \label{sec:var_rate}

Clearly, any supply allocation $A(t)$ that respects \eqref{eq:var_rate2} also satisfies \eqref{eq:var_rate1}. 

To prove the converse,  first assume $r>0$. Consider an $A(t)$ satisfying \eqref{eq:var_rate1}. Then, $A(t) > 0$ for at least $k+1$ time slots and $A(t) = m$ for at most $k$ time slots. Pick $k+1$ time slots with largest value of $A(t)$. Define $A^1(t) =1$ at the selected slots and $0$ otherwise. Let $B(t) = A(t) - A^1(t)$. Then, $B(t) \in \{0,1,\ldots,m-1\}$ and $\sum_t B(t) = k(m-1) +(r-1)$. If $r=0$, then $A(t) > 0$ for at least $k$ time slots and $A(t) = m$ for at most $k$ time slots. Pick $k$ time slots with largest value of $A(t)$. Define $A^1(t) =1$ at the selected slots and $0$ otherwise. Let $B(t) = A(t) - A^1(t)$. Then, $B(t) \in \{0,1,\ldots,m-1\}$ and $\sum_t B(t) = k(m-1) $. Thus, we can always write $A(t)$ as
\begin{equation}
A(t) = A^1_t + B(t),
\end{equation}
where $A^1_t \in \{0,1\}$, $\sum_t A^1_t = k + \mathds{1}_{\{r >0\}}$, $B(t) \in \{0,1,\ldots,m-1\}$ and $\sum_t B(t) = k(m-1) + (r-1)^+$. Now, if $(r-1)>0$, then $B(t) >0$ for at least $k+1$ slots and $B(t) = m-1$ for at most $k$ slots. Pick $k+1$ time slots with largest value of $B(t)$. Define $A^2(t) =1$ at the selected slots and $0$ otherwise. Let $C(t) = B(t) - A^2(t)$. Then, $C(t) \in \{0,1,\ldots,m-2\}$ and $\sum_t B(t) = k(m-2) +(r-2)$. If $(r-1)^+ =0$, then $B(t) > 0$ for at least $k$ time slots and $B(t) = m$ for at most $k$ time slots. Pick $k$ time slots with largest value of $B(t)$. Define $A^2(t) =1$ at the selected slots and $0$ otherwise. Let $C(t) = B(t) - A^2(t)$. Then, $C(t) \in \{0,1,\ldots,m-2\}$ and $\sum_t C(t) = k(m-2) $. Thus, we can write $A(t)$ as
\begin{equation}
A(t) = A^1_t + A^2_t +  C(t),
\end{equation}
where $A^2_t \in \{0,1\}$, $\sum_t A^2_t = k + \mathds{1}_{\{r >1\}}$,$C(t) \in \{0,1,\ldots,m-2\}$ and $\sum_t C(t) = k(m-2) + (r-2)^+$. Continuing sequentially, we can decompose $A(t)$ into $A^1(t),\ldots, A^m(t)$.

\bibliographystyle{IEEEtran}
\bibliography{timedif}

\end{document}